\documentclass[journal]{IEEEtran}
\IEEEoverridecommandlockouts
\usepackage[dvipsnames]{xcolor}

\usepackage{cite}
\usepackage{amsmath,graphicx}
\usepackage{amssymb,amsfonts}
\usepackage{graphicx}
\usepackage{textcomp}
\usepackage{xcolor}
\usepackage{subfigure}
\usepackage{textcomp} 
\usepackage{siunitx}  
\usepackage{algorithm}
\usepackage{algpseudocode}
\usepackage{xpatch}
\usepackage{pgfplots}
\usepackage{pgfplotstable}
\usepackage{readarray}
\usepackage{siunitx}
\usepackage{bm}
\usepackage[nolist]{acronym}
\usepackage{bbm}
\usepackage{url}
\usepackage[english]{babel}
\usepackage{amsthm}
\usepackage{amssymb}

\usepackage{needspace}

\usepackage{tikz}
\usetikzlibrary{shapes.geometric}

\usetikzlibrary{spy}
\allowdisplaybreaks

\definecolor{TUMBeamerYellow}    {rgb} {1.000,0.706,0.000}    
\definecolor{TUMBeamerOrange}    {rgb} {1.000,0.502,0.000}    
\definecolor{TUMBeamerRed}       {rgb} {0.898,0.204,0.094}    
\definecolor{TUMBeamerDarkRed}   {rgb} {0.792,0.129,0.247}    
\definecolor{TUMBeamerBlue}      {rgb} {0.000,0.600,1.000}    
\definecolor{TUMBeamerLightBlue} {rgb} {0.255,0.745,1.000}    
\definecolor{TUMBeamerGreen}     {rgb} {0.569,0.675,0.420}    
\definecolor{TUMBeamerLightGreen}{rgb} {0.710,0.792,0.510}    
\definecolor{TUMBlack}           {rgb} {0.000,0.000,0.000}    
\definecolor{TUMWhite}           {rgb} {1.000,1.000,1.000}    
\definecolor{TUMBlue}            {rgb} {0.000,0.396,0.741}    
\definecolor{TUMDarkBlue}        {rgb} {0.000,0.200,0.349}    
\definecolor{TUMDarkerBlue}      {rgb} {0.000,0.322,0.576}    
\definecolor{TUMMediumBlue}      {rgb} {0.000,0.451,0.812}    
\definecolor{TUMLighterBlue}     {rgb} {0.392,0.627,0.784}    
\definecolor{TUMLightBlue}       {rgb} {0.596,0.776,0.918}    
\definecolor{TUMDarkGray}        {rgb} {0.345,0.345,0.353}    
\definecolor{TUMMediumGray}      {rgb} {0.612,0.616,0.624}    
\definecolor{TUMLightGray}       {rgb} {0.851,0.855,0.859}    
\definecolor{TUMGreen}           {rgb} {0.635,0.678,0.000}    
\definecolor{TUMOrange}          {rgb} {0.890,0.447,0.133}    
\definecolor{TUMIvory}           {rgb} {0.855,0.843,0.796}    
\DeclareMathOperator{\tr}{tr}
\DeclareMathOperator{\T}{T}
\DeclareMathOperator{\He}{H}

\DeclareMathOperator{\inver}{-1}

\DeclareMathOperator{\diag}{diag}
\DeclareMathOperator{\rank}{rank}
\DeclareMathOperator{\imag}{j\hspace*{-2pt}}

\newcommand{\inv}{{\inver}}
\newcommand{\eye}{\bm{\mathrm{I}}}
\newcommand{\bbthet}{\bm{\bar{\theta}}}
\newcommand{\bD}{\bm{D}}
\newcommand{\Her}{{\He}}
\newcommand{\td}{{\text{d}}}

\newcommand{\tdk}{{\text{d},k}}
\newcommand{\tc}{{\text{c}}}
\newcommand{\tck}{{\text{c},k}}

\newcommand{\bC}{\bm{C}}
\newcommand{\bQ}{\bm{Q}}

\newcommand{\bV}{\bm{V}}
\newcommand{\NB}{N_{\text{B}}}

\newcommand{\NRe}{N_{\text{R}}}

\newcommand{\bU}{\bm{U}}
\newcommand{\tre}{{\text{r}}}

\newcommand{\trek}{{\text{r},k}}
\newcommand{\ts}{{\text{s}}}
\newcommand{\bH}{\bm{H}}

\newcommand{\bh}{\bm{h}}
\newcommand{\mThet}{\bm{\Theta}}
\newcommand{\bu}{\bm{u}}

\newcommand{\bv}{\bm{v}}
\newcommand{\be}{\bm{e}}
\newcommand{\transpo}{{\T}}

\newcommand{\bthet}{\bm{\theta}}

\newcommand{\cmplx}[1]{\mathbb{C}^{#1}}
\newcommand{\norm}[1]{\|#1\|}
\newcommand{\abs}[1]{\left|#1\right|}

\newcommand{\expct}[1]{\mathbb{E}\hspace*{-2pt}\left[#1\right]}

\newcommand{\gaussdist}[2]{\mathcal{N}_{\mathbb{C}}(#1,#2)}
\newcommand{\summe}[2]{\sum_{#1}^{#2}}

\def\BibTeX{{\rm B\kern-.05em{\sc i\kern-.025em b}\kern-.08em
    T\kern-.1667em\lower.7ex\hbox{E}\kern-.125emX}}

\newtheorem{theorem}{Theorem}
\newtheorem{lemma}{Lemma}
\newtheorem{proposition}{Proposition}
\newtheorem{corollary}{Corollary}

\begin{document}

\begin{acronym}
    \acro{AoD}{angle of departure}
    \acro{AoA}{angle of arrival}
    \acro{ULA}{uniform linear array}
    \acro{CSI}{channel state information}
    \acro{LOS}{line-of-sight}
    \acro{EVD}{eigenvalue decomposition}
    \acro{BS}{base station}
    \acro{MS}{mobile station}
    \acro{mmWave}{millimeter wave}
    \acro{DPC}{dirty paper coding}
    \acro{RIS}{reconfigurable intelligent surface}
    \acro{AWGN}{additive white gaussian noise}
    \acro{MIMO}{multiple-input multiple-output}
    \acro{UL}{uplink}
    \acro{DL}{downlink}
    \acro{OFDM}{orthogonal frequency-division multiplexing}
    \acro{TDD}{time-division duplex}
    \acro{LS}{least squares}
    \acro{MMSE}{minimum mean square error}
    \acro{SINR}{signal-to-interference-plus-noise ratio}
    \acro{OBP}{optimal bilinear precoder}
    \acro{LMMSE}{linear minimum mean square error}
    \acro{MRT}{maximum ratio transmitting}
    \acro{M-OBP}{multi-cell optimal bilinear precoder}
    \acro{S-OBP}{single-cell optimal bilinear precoder}
    \acro{SNR}{signal-to-noise ratio}
    \acro{SDR}{Semidefinite Relaxation}
    \acro{SE}{spectral efficiency}
    \acro{GCEs}{Gram channel eigenvalues}
    \acro{BCD}{block coordinate descent}
    \acro{LISA}{linear successive allocation}
    \acro{WMMSE}{weighted minimum mean square error}
    \acro{SVD}{singular value decomposition}
    \acro{STAR}{simultaneously transmitting and reflecting}
    \acro{BD}{beyond diagonal}
\end{acronym}


\title{High-SNR Comparison of Linear Precoding\\and DPC in RIS-Aided MIMO Broadcast Channels}

\author{\IEEEauthorblockN{Dominik Semmler,  Benedikt Fesl, Michael Joham, and Wolfgang Utschick\\}
\IEEEauthorblockA{\textit{School of Computation, Information and Technology, Technical University of Munich, Germany} \\
email: \{dominik.semmler,benedikt.fesl,joham,utschick\}@tum.de}
}

\maketitle

\begin{abstract}
    We compare \ac{DPC} and linear precoding methods in a \ac{RIS}-aided high-\ac{SNR} scenario, where the channel between the \ac{BS} and the \ac{RIS} is dominated by a \ac{LOS} component.
    Furthermore, we consider two groups of users where one group can be efficiently served by the \ac{BS}, whereas the other one has a negligible direct channel and has to be served via the \ac{RIS}.
    Within this scenario, we analytically show fundamental differences between \ac{DPC} and linear methods.
    In particular, our analysis addresses two essential aspects, i.e., the orthogonality of the \ac{BS}-\ac{RIS} channel with the direct channel and a channel mitigation term, depending on the number of \ac{RIS} elements, that is present only for linear precoding techniques.
    The mitigation term generally leads to strong limitations for the linear method, especially for random or statistical phase shifts.
    Moreover, we discuss under which circumstances this mitigation term is negligible and in which scenarios \ac{DPC} and linear precoding lead to the same performance.
\end{abstract}
\begin{figure}[b]
    \onecolumn
    \centering
    \scriptsize{This work has been submitted to the IEEE for possible publication. Copyright may be transferred without notice, after which this version may no longer be accessible.}
    \vspace{-1.3cm}
    \twocolumn
\end{figure}
\section{Introduction}
\acp{RIS} are considered to be an important technology for future wireless communication systems which are able to significantly enhance the system performance \cite{Power_Min_IRS}.
The \ac{RIS} is a surface consisting of many passive elements which allow to reconfigure the channel environment.
Modelling the \ac{RIS} is still under ongoing research.
Various architectures for \acp{RIS} exist, e.g., \ac{STAR} \acp{RIS} (see \cite{STARRIS,STARRISTwo}),  \ac{BD}-\acp{RIS} where the reflecting elements are connected to each other (see \cite{BeyondDiagonalScattering}, \cite{BeyondDiagonal}), as well as active \acp{RIS} (see \cite{ActiveVSPassive}, \cite{ActiveVSPassiveTwo}).
The conventional architecture is the passive, single-connected \ac{RIS}, where the elements are not connected with each other.
This architecture is usually modeled with the phase shift model for which a considerable amount of literature is available (see, e.g., \cite{Power_Min_IRS,EnergyEff,WSR,WMMSEMIMO,MIMOP2PCap,StatisticalCSITwo}).
Models based on the impedance and scattering formulations (see \cite{BeyondDiagonalScattering,Impedance}) have been compared to the conventional phase-shift model in \cite{ScatteringImpedanceCompare,FollowUPSZJournal,FollowUpUniversal,WSAComparison}.
These models also allow to incorporate mutual coupling (see \cite{Impedance,MutualCoupling,DecouplingNetwork}) between the reflecting elements in the system model.
Mutual coupling can, however, also be taken into account with decoupling networks (see \cite{DecouplingNetwork}) which have the beneficial aspect that  the resulting system model has the same structure as in the case without mutual coupling. 
Throughout this article, we only consider the conventional phase shift model without mutual coupling.
However, the results of this article can be directly extended to the model in \cite{ScatteringImpedanceCompare} as well as to mutual coupling with the help of decoupling networks (see \cite{DecouplingNetwork}).

A major challenge for the \ac{RIS} design is the large number of reflecting elements which have to be jointly optimized with the transmit filters in a downlink scenario.
The classical approach is to first estimate the channel and, afterwards, solve the joint optimization of precoders and reflecting elements.
When considering perfect estimates, and, hence, assuming perfect instantaneous \ac{CSI}, the \ac{RIS} already showed a significant performance increase in various scenarios, e.g., for power consumption \cite{Power_Min_IRS}, energy efficiency \cite{EnergyEff}, as well as for spectral efficiency \cite{WSR,WMMSEMIMO,MIMOP2PCap}.
However, channel estimation in a \ac{RIS} scenario (see \cite{ChannelEstimation,ChannelEstimationTwo}) is difficult due to the high number of parameters.
While there are methods which approach this problem (see \cite{ChannelEstimationReducedPilots}, \cite{ChannelEstimationReducedPilotsTwo}), designing the reflecting elements only based on the statistics of the channel (see \cite{StatisticalCSITwoTimeScale,StatisticalCSI,StatisticalCSITwo,StatisticalCSIThree}) circumvents the channel estimation phase.
Using only statistical \ac{CSI} does not require updating the reflecting elements in each coherence interval of the channel which is an additional upside.
In this article, we mathematically analyze the \ac{SE} when the reflecting elements are optimized based on instantaneous \ac{CSI}, statistical \ac{CSI} as well as when they are chosen randomly.

When it comes to the precoder design at the base station, we consider both \ac{DPC} (see \cite{DPC}, \cite{DPCTwo}) as well as linear approaches.
For the linear methods, we use zero-forcing beamforming which is high-\ac{SNR} optimal.
\ac{DPC} has already been compared to zero-forcing in the conventional channel model without the \ac{RIS} (see \cite{ZFDPCCompare}, \cite{HighSNRInstant}).
\ac{DPC}, being the optimal choice for the downlink scneario, always leads to a better \ac{SE} which is especially noticeable when the channel is ill-conditioned.
In this article, the difference between \ac{DPC} and linear precoding is analyzed mathematically and we show that their difference is generally significant when including an \ac{RIS}. 

An important aspect throughout this article is that we are considering a scenario in which the channel between the \ac{BS} and the \ac{RIS} is given by a rank-one matrix, motivated by the fact that the \ac{BS} and the \ac{RIS} are typically deployed in \ac{LOS} and, additionally, at a considerable height (see \cite{LOSAssump}).
This has already been assumed in a number of articles (see, e.g., \cite{MaxMinSINR,LOSAssumpOne,LOSAssumpTwo,LOSAssumpThree,LOSAssumpFour}).
Under this assumption, the eigenvalue result of \cite{Eigenvalues} holds, i.e., the eigenvalues of the composite channel Gram matrix (we assume that all user channels are stacked into this matrix), including the \ac{RIS}, can, at maximum, be improved to the next larger eigenvalues of the composite direct channel Gram matrix.
While this result shows clear limitations, the \ac{RIS} can still have a significant impact on the scenario (see \cite{LOSZeroForc}), especially for a lower number of users.

Furthermore, we consider a scenario where the users can be divided into two groups.
One group consists of users with a strong direct channel, which could be efficiently served by the \ac{BS} without the \ac{RIS},
whereas the other group consists of users with a negligible direct channel and have to be served via the \ac{RIS}.
In this case, the eigenvalues of the weak users' direct channels can be modeled as being approximately zero, and, according to \cite{Eigenvalues}, only one eigenvalue can be improved due to the rank-one assumption for the \ac{BS}-\ac{RIS} channel.
Therefore, only one user of the weak user group will be additionally served. 
Accordingly, throughout this article, we assume w.l.o.g. that the number of weak users to be served via the \ac{RIS} is one. 

If the strong users did not exist and only the weak users were present, the direct composite Gram channel matrix has only one non-zero eigenvalue and would be rank-one.
Consequently, the optimal solution for sum-\ac{SE} maximization is to only allocate one user in total.
This scenario was already analyzed in \cite{MaxMinSINR} under the max-min fairness criterion in an asymptotic regime,
where it was also proposed to only serve a single user.
In contrast, when strong users do exist and the direct channel is non-negligible, the situation is different and it is not necessarily sum-\ac{SE} optimal to just improve the weak user's channel gain as the additional strong users introduce interference.

In this article, this scenario, with some direct channels being non-negligible, is analyzed in the high-\ac{SNR} regime.
For high \ac{SNR} values, all the strong direct channel users are allocated as well as one additional user of the weak users group via the \ac{RIS}.
Hence, we obtain a rank-improvement scenario for which the \ac{RIS} is particularly suited (see  \cite{RankImprov}).

In particular, we make the following contributions:
\begin{itemize}
    \item We derive the optimal \ac{SE} expressions for linear precoding and \ac{DPC} in the high-\ac{SNR} regime.
          These expressions allow us to analyze this scenario mathematically. Additionally, we can show that fundamental differences between \ac{DPC} and linear precoding occur also in a \ac{RIS}-assisted scenario.
    \item We show analytically that linear precoding suffers from a mitigation term that completely vanishes in the case of \ac{DPC}.
            Moreover, we show the importance of the \ac{BS}-\ac{RIS} channel being orthogonal to the direct channel and that \ac{DPC} has an advantage if this condition is not perfectly fulfilled.
    \item More results are derived from the considerations above.
            Firstly, only for \ac{DPC}, it is high-\ac{SNR} optimal to maximize the weak user's channel gain, whereas for linear precoding, the solution is more intricate and has no closed-form.
          Secondly, when considering random or statistical phase shifts, the \ac{SE} for linear precoding saturates for an increasing number of \ac{RIS} elements in case of i.i.d. Rayleigh fading, whereas for \ac{DPC}, the \ac{SE} is increasing monotonically.
\end{itemize}
\section{System Model}
A \ac{DL} scenario is considered in which $K+1$ single-antenna users are served by a single \ac{BS} having $\NB$ antennas.
Additionally, we assume one \ac{RIS} having $\NRe$ reflecting elements.
The channel from the \ac{BS} to the $k$-th user reads as
\begin{equation}
    \bh_k^\Her = \bh^\Her_\tdk + \bh_\trek^\Her \mThet \bm{G} \in \cmplx{1 \times \NB}
\end{equation}
where $\bh_\tdk^\Her \in \cmplx{1 \times \NB}$ is the direct channel from the \ac{BS} to the $k$-th user,  $ \bh_\trek^\Her \in \cmplx{1 \times \NRe}$ is the reflecting channel from the \ac{RIS}
to the $k$-th user, and
\begin{equation}
    \bm{G} = \sqrt{L_{\text{G}} \NB} \bm{a}\bm{b}^\Her \in \cmplx{\NRe \times \NB}
\end{equation} is the \ac{LOS} rank-one channel from the \ac{BS} to the \ac{RIS} where $L_{\text{G}}$ is the pathloss.
We have a scaling with $\NB$ as the steering vector $\bm{b}$  is normalized such that $\norm{\bm{b}}_2=1$, whereas $\norm{\bm{a}}_2^2 = \NRe$.
The matrix $\mThet = \diag(\bthet) \in \cmplx{\NRe \times \NRe}$ with $\bthet \in \{\bm{z} \in \cmplx{\NRe}: z_n = e^{ \imag\phi_n}, \; \phi_n \in [0,2\pi),\; \forall n\}$ is the phase manipulation at the \ac{RIS}.
Recently, new models appeared based on the impedance formulation (see \cite{ScatteringImpedanceCompare}) where the constraints are of the form $\{\bm{z} \in \cmplx{\NRe}: z_n = e^{ \imag\phi_n} - 1, \; \phi_n \in [0,2\pi),\; \forall n\}$.
While we consider the unit-modulus constraint set from above, all the results from this article can be directly transferred to the other models.
Additionally, we neglect mutual coupling in this article and, hence, the phase shift matrix is diagonal.
However, when using decoupling networks (see \cite{DecouplingNetwork}), the results can be directly extended to the case of mutual coupling as the phase shift matrix becomes diagonal again in this case.
Even for non-diagonal $\mThet$ matrices, e.g., in case of mutual coupling without decoupling networks, or in the case of \ac{BD}-\acp{RIS}, similar expressions as in this article can be obtained. 

In case $\mThet$ is diagonal, a more compact notation is available and is used for notational convenience throughout this article.
In particular, we rewrite $\bh_k^\Her$ as
\begin{equation}\label{eq:Cascaded}
    \bh_k^\Her = \bh^\Her_\tdk + \bh_\tck^\Her \bthet \bm{b}^\Her \in \cmplx{1 \times \NB}
\end{equation}
with the cascaded channel being defined as
\begin{align}
    \bh_\tck^\Her &= \bh_\trek^\Her \diag(\bm{a})\sqrt{L_G \NB}.
\end{align}

As pointed out in the introduction, we assume that exactly one of the users has a negligible direct channel and that this user can only be properly served via the \ac{RIS}. 
We consider user $K+1$ to be this weak user, whereas the other, i.e., strong users ($1$ until $K$) are stacked into the channel matrices 
$\bm{H}_{\text{d}}^\ts =[\bh_{\td,1},\dots,\bh_{\td,K}]^\Her$, $\bm{H}_{\text{r}}^\ts =[\bh_{\tre,1},\dots,\bh_{\tre,K}]^\Her$ and according to the notation in \eqref{eq:Cascaded} equivalently $\bm{H}_{\tc}^\ts =[\bh_{\tc,1},\dots,\bh_{\tc,K}]^\Her$.
Assuming that only one user has a negligible direct channel is motivated by the fact that when having a group of weak users with no direct channel, only one can be additionally allocated via the \ac{RIS}.
This is because the \ac{BS}-\ac{RIS} channel is rank-one and this fact has been discussed in the introduction based on the eigenvalue result of \cite{Eigenvalues}.
However, this can also be directly observed when considering the composite channel matrix 
\begin{equation}
    \bH = { \begin{bmatrix}
        \bm{H}_{\text{d}}^\ts\\
        \bm{0}\\
    \end{bmatrix}}
    +
\bm{H}_{\text{c}} \bthet \bm{b}^\Her
\end{equation}
where $\bH_\tc$ are the stacked cascaded channels of all users.
We can directly see that the maximum rank of the composite channel matrix is $\rank(\bH) \le \rank(\bH_\td^\ts)+1 = K+1$ and, therefore, maximally one additional user is allocated.
Hence, w.l.o.g we can write the composite channel matrix as
\begin{equation}
    \label{eq:CompositeChannelMatrixDefinition}
    \bm{H}_{} = \begin{bmatrix}
        \bH^\ts\\
        \bh_{K+1}^\Her
    \end{bmatrix}     =\underbrace{ \begin{bmatrix}
        \bm{H}_{\text{d}}^\ts\\
        \bm{0}^\Her\\
    \end{bmatrix}}_{\bm{H}_{\text{d}}}
    +
    \underbrace{ \begin{bmatrix} 
        \bm{H}_{\text{c}}^\ts\\
        \bh_{\tc,K+1}^\Her
    \end{bmatrix} }_{\bm{H}_{\text{c}}}
    \bthet\bm{b}^\Her \in \cmplx{K+1\times \NB}.
\end{equation}
\section{SE Expressions for DPC and ZF}
In this section, we derive the asymptotic \ac{SE} expressions at high-\ac{SNR} for both linear precoding and \ac{DPC}.
To this end, we exploit the rank-one property of the \ac{BS}-\ac{RIS} channel.
Firstly, we split the channel matrix $\bH$ in the subspace via the \ac{BS}-\ac{RIS} channel $\bm{b}\bm{b}^\Her$ and its orthogonal space $ \bm{P}_{\bm{b}}^{\perp}$ as 
\begin{equation}
    \bH = \bH_\td \bm{P}_{\bm{b}}^{\perp} + \left(\bm{H}_{\text{c}}\bthet + \bH_\td\bm{b} \right)\bm{b}^\Her
\end{equation}
where $ \bm{P}_{\bm{b}}^{\perp}$ is the orthogonal projector
\begin{equation}
    \bm{P}_{\bm{b}}^{\perp} = \eye - \bm{b}\bm{b}^\Her.
\end{equation}
Similar to \cite{Eigenvalues}, this allows us to express the channel Gram matrix as
\begin{align}
    \label{eq:HHDttDDef}
    \bm{H}_{}\bm{H}_{}^{\He} &= \begin{bmatrix}
        \bm{C}_\ts & \bm{0}\\
        \bm{0}^\transpo& 0
    \end{bmatrix} + \bm{D} \bbthet \bbthet^\Her \bm{D}^\Her
\end{align}
with the definitions
    \begin{align}
        \label{eq:CiDefinition}
    \bm{C}_\ts  &= \bm{H}_{\text{d}}^\ts\bm{P}_{\bm{b}}^{\perp} \bm{H}_{\text{d}}^{\ts,\Her} \in \cmplx{K \times K}, \; \\
    \label{eq:DiDefinition}
      \bD &= \begin{bmatrix}
        \bm{H}_{\text{c}},\; \bm{H}_{\text{d}}\bm{b}
    \end{bmatrix},
    \quad  \bbthet = \begin{bmatrix}
        \bthet^\Her & 1
    \end{bmatrix}^\Her.
\end{align}
The matrix $\bC_\ts$ corresponds to the direct Gram channel matrix exluding the direction $\bm{b}$ to the \ac{RIS},
whereas $\bD \bbthet$ represents the channel via the direction $\bm{b}$, i.e., the direction via the \ac{RIS}.
Note, that also the direct channel component $\bH_\td \bm{b}$ is incorporated in the part $\bD \bbthet$.
In the following, we assume that $\be_{K+1}^{\T}\bD\bbthet = \bh_{\text{c},K+1}^\Her\bthet \neq 0$ where $\be_n$ is the $n$-th canonical basis vector.
This means that the weak user $K+1$ has a non-zero \ac{RIS}-user channel.
For example, in case no \ac{RIS} were present, i.e., $\bthet = \bm{0}$, the weak user would effectively not exist.
\subsection{Linear Precoding}
For linear precoding, it is known that zero-forcing is optimal at high-\ac{SNR} and the sum-\ac{SE} is given by (see \cite{HighSNRInstant})
\begin{equation}
    \label{eq:ActualOptProblem}
    \begin{aligned}
    \ac{SE}_{\text{ZF-General}}=\summe{k=1}{K} \log_2\left(1+ \frac{\gamma_{k}}{\be^\transpo_k (\bH \bH^\Her)^\inv\be_k}\right)
    \end{aligned}
   \end{equation}
where $\gamma_k$ are the power allocations.
Under the condition $\bh_{\text{c},K+1}^\Her\bthet \neq 0$ and by defining
\begin{equation}\label{eq:DStrongDef}
    \bD_\ts =[\bm{H}^\ts_{\text{c}},\;\bm{H}^\ts_{\text{d}}\bm{b}]
\end{equation}
as the channel component via the direction $\bm{b}$ of only the strong users, the inverted channel gain for the $k$-th user can be written (see Appendix \ref{app:ZFEffectiveGains}) as
     \begin{equation}\label{eq:ZFEffectiveGains}
        \bm{e}_k^{\T}\left(\bH \bH^\Her \right)^{\inv}\bm{e}_k =
        \begin{cases}
            \bm{e}_{k}^{\T}  \bm{C}_\ts^{\inv}  \bm{e}_{k}  \quad &\text{if} \quad k \le K\vspace*{0.2cm}\\
            \frac{1+\bbthet^{\Her} \bD_\ts^{\Her}\bm{C}_\ts^{\inv} \bD_\ts \bbthet}{\abs{\bh_{\tc,K+1}^\Her \bthet}^2} \quad &\text{if} \quad k=K+1.
        \end{cases}
     \end{equation}
     Considering a uniform power allocation $\gamma_k = \frac{P_{\text{Tx}}}{K} = \bar{p}$,
     which is high-\ac{SNR} optimal, we obtain the sum-\ac{SE}
     \begin{equation}\label{eq:SEZF}
        \begin{aligned}
       \text{SE}_{\text{ZF}} &=  \hspace*{-3pt}\summe{k =1}{K}\log_2\left(\hspace*{-2pt}1+\hspace*{-1pt}\hspace*{-1pt}\frac{\bar{p}}{\bm{e}_k^{\T} \bC_\ts^{\inv} \bm{e}_k} \right)\\
       &+ \log_2\left(\hspace*{-2pt}1\hspace*{-1pt}+\hspace*{-1pt}\frac{\bar{p}\abs{\bh_{\tc,K+1}^\Her \bthet}^2}{1+\bbthet^{\Her} \bD_\ts^{\Her}\bm{C}_\ts^\inv\bD_\ts \bbthet }\right)\hspace*{-2pt}.
        \end{aligned}
    \end{equation}
\vspace*{-1.cm}    
\subsection{DPC}
    For \ac{DPC}, we use the dual-\ac{UL} representation 
    \begin{equation}
        \text{SE}_{\text{DPC-General}} = \log_2\det\left(\eye + \bH^\Her\bQ \bH \right)
    \end{equation}
    where choosing the transmit covariance matrix
       $\bQ = \eye \frac{P_{\text{Tx}}}{K} = \eye\bar{p}$
    is high-\ac{SNR} optimal.
    Defining $\lambda_k$ as the eigenvalues of $\bC_\ts$ in decreasing order with the corresponding eigenvectors $\bm{u}_k$, the sum-\ac{SE} for \ac{DPC} is given (see Appendix \ref{app:DPCAsymSE}) by
    \begin{equation}\label{eq:SEDPC}
        \begin{aligned}
           & \text{SE}_{\text{DPC}} =  \summe{k=1}{K}\log_2\left(1 + \lambda_k \bar{p}\right)\\
            & +  \log_2\left(1 + \abs{\bh_{\tc,K+1}^\Her \bthet}^2 \bar{p}+ \hspace*{-3pt}\summe{k=1}{K}\abs{\bbthet^\Her\bD^\Her \bu_k}^2  \hspace*{-3pt} \frac{ \bar{p}}{1+\lambda_k  \bar{p}} \right).\\
        \end{aligned}
    \end{equation}
Having derived the sum-\ac{SE} expressions for linear precoding and \ac{DPC} in \eqref{eq:SEZF} and \eqref{eq:SEDPC}, we will now analyze the implications arising from these two formulations.
It is important to note that both expressions are optimal in the high-\ac{SNR} regime.

\section{High-SNR Analysis of DPC and Linear Precoding}
In the following, we are considering the high-\ac{SNR} regime (for which \eqref{eq:SEZF} and \eqref{eq:SEDPC} are optimal).
Hence, we have \mbox{$\bar{p} \rightarrow \infty$} and obtain the asymptotic expressions of \eqref{eq:SEZF} and \eqref{eq:SEDPC}
\begin{align}
    \label{eq:HSNRZF}
    {\overline{\ac{SE}}}_{\text{Lin}} &= \underbrace{\hspace*{-2pt} \summe{k =1}{K}\log_2\hspace*{-2pt} \left(\frac{\bar{p}}{\bm{e}_k^{\T} \bC_\ts^{\inv} \bm{e}_k} \right)}_{{\overline{\text{SE}}}_{\text{Lin},\td}} 
    \hspace*{-2pt} +\hspace*{-2pt}  \underbrace{\log_2\hspace*{-2pt} \left(\frac{\abs{\bh_{\tc,K+1}^\Her \bthet}^2\bar{p}}{1+\bbthet^{\Her} \bD_\ts^{\Her}\bm{C}_\ts^\inv\bD_\ts \bbthet }\right)}_{\overline{\text{SE}}_{\text{Lin},\tre}},\\
    \label{eq:HSNRDPC}
    {\overline{\ac{SE}}}_{\text{DPC}} &= \underbrace{\log_2\det\left(\bC_{\ts}\bar{p} \right)}_{{\overline{\text{SE}}}_{\text{DPC},\td}} +  \underbrace{\log_2\left(  \abs{\bh_{\tc,K+1}^\Her  \bthet}^2 \bar{p} \right)}_{\overline{\text{SE}}_{\text{DPC},\tre}}
\end{align}
for linear precoding and \ac{DPC}, respectively.
\ac{DPC} is the capacity achieving method, and after a reformulation of \eqref{eq:HSNRDPC}, an important observation can be made.
\needspace{3\baselineskip}
\begin{proposition}
     In the high-\ac{SNR} regime, the \ac{SE} of \ac{DPC} can  be given as
\begin{align}
   \nonumber {\overline{\ac{SE}}}_{\text{DPC}} &= \log_2\det(\bH_\td^\ts \bH_\td^{\ts,\Her} \bar{p})\\
    & + \log_2(\bm{b}^\Her \bm{P}_{{\bH_\ts^\Her}}^{\perp}\bm{b}) + \log_2(\abs{\bh_{\tc,K+1}^\Her \bthet}^2 \bar{p} ).
\end{align}
\end{proposition}
\begin{proof}
    This result can be directly obtained from \eqref{eq:HSNRDPC} by rewriting $\log_2\det(\bC_\ts) = \log_2\det(\bH_\td^\ts \bH_\td^{\ts,\Her} - \bH_\td^\ts \bm{b}\bm{b}^\Her\bH_\td^{\ts,\Her} ) = \log_2\det(\bH_\td^\ts \bH_\td^{\ts,\Her}) + \log_2(1- \bm{b}^\Her\bH_\td^{\ts,\Her} (\bH_\td^{\ts} \bH_\td^{\ts,\Her})^{\inv}\bH_\td^\ts \bm{b})$.
\end{proof}
This reformulation is important as the expression matches the intuition in the sense that the users with non-negligible direct channels are served via the \ac{BS} corresponding to the term $\log_2\det(\bH_\td^\ts \bH_\td^{\ts,\Her} \bar{p})$.
Furthermore, when including the \ac{RIS}, one additional user with negligible direct channel is served via the \ac{RIS} corresponding to the term $\log_2(\abs{\bh_{\tc,K+1}^\Her \bthet}^2 \bar{p} )$.
However, we further have the term $\log_2(\bm{b}^\Her \bm{P}_{\bH_\td^\ts}^{\perp}\bm{b})$, which highlights the importance of the \ac{BS}-\ac{RIS} channel being orthogonal to the direct channels.
This is an important observation which is analyzed in more detail in Section \ref{subsec:OrthGHd}.

The \ac{SE} for linear precoding leads to a fundamentally different structure.
However, we can see that the expression in \eqref{eq:HSNRZF} still has a special structure in the sense that the strong users $k=1,\dots,K$ do not depend on the reflecting elements which will be discussed in the following.

It is known that zero-forcing precoding is optimal in the high-\ac{SNR} regime.
Hence, the interference has to be completely cancelled by precoding at the \ac{BS}.
Considering user $K+1$, the corresponding channel is given by $\bh_{K+1} = \bh_{\tc,K+1}^\Her \bthet \bm{b}^\Her \in \cmplx{1 \times \NB}$ as this user has no direct channel.
To achieve interference cancellation, all precoders of the strong users $1,2,\dots,K$ have to be orthogonal to the channel $\bh_{K+1}$, i.e., the direction $\bm{b}$ via the \ac{RIS}.
Hence, the strong users are spatially separated and we obtain the structure in \eqref{eq:HSNRZF} as all precoders for the users $1,2,\dots,K$ are chosen orthogonal to the channel via the \ac{RIS}.
From this spatial separation, it additionlly follows that the channel part via the direction $\bm{b}$, i.e., $\bH_\td^\ts\bm{b}\bm{b}^\Her$ and $\bH_\tc^\ts \bthet \bm{b}^\Her$, is decremental for the strong users.
This is taken care of by the precoder of user $K+1$ and the negative impact is expressed in the channel mitigation term $\bbthet^{\Her} \bD_\ts^{\Her}\bm{C}_\ts^\inv\bD_\ts \bbthet$.
This problem does not arise for \ac{DPC} as it can additionally cancel the interference by coding.
The scenario is graphically illustrated in Fig. \ref{fig:Szenario}.
Please note that the figure is a theoretical illustration according to the notation of \eqref{eq:Cascaded} and the physical channel from the \ac{BS} to the \ac{RIS} is the ${\NRe \times \NB}$ dimensional channel $\bm{G}$ and not the direction $\bm{b}^\Her \in \cmplx{1\times\NB}$.

Hence, fundamental differences betweeen \ac{DPC} and linear precoding arise which are stated in the following theorem.
\needspace{3\baselineskip}
\begin{theorem}
    \mbox{In the high-\ac{SNR}} regime, the difference between \ac{DPC} and linear precoding can be given by 
    \begin{equation}
        \Delta\overline{\text{SE}} = {\overline{\ac{SE}}}_{\text{DPC}}-{\overline{\ac{SE}}}_{\text{Lin}}
    \end{equation}
which can be further written as 
\begin{align}
    \Delta\overline{\text{SE}} = \underbrace{{\overline{\text{SE}}}_{\text{DPC},\td} - {\overline{\text{SE}}}_{\text{Lin},\td}}_{\Delta{\overline{\text{SE}}}_{\td}} + \underbrace{{\overline{\text{SE}}}_{\text{DPC},\tre} - {\overline{\text{SE}}}_{\text{Lin},\tre}}_{\Delta{\overline{\text{SE}}}_{\tre}}
\end{align}
with the expressions
\begin{align}\label{eq:DeltadRate}
        \Delta{\overline{\text{SE}}}_{\td} &= \log_2\det\left(\bC_{\ts} \right) - \summe{k =1}{K}\log_2\left( (\bm{e}_k^{\T} \bC_\ts^{\inv} \bm{e}_k)^{\inv} \right),\\
        \label{eq:DeltarRate} \Delta{\overline{\text{SE}}}_{\tre} &= \log_2(1+\bbthet^{\Her} \bD_\ts^{\Her}\bm{C}_\ts^\inv\bD_\ts \bbthet ).
    \end{align}
\end{theorem}
The discussion of this theorem is separated into Section \ref{subsec:OrthGHd} for \eqref{eq:DeltadRate} and into Section \ref{subsubsec:InflRefl} for \eqref{eq:DeltarRate}.

\vspace*{0.5cm}
\begin{figure}[t!]
    \centering
        \includegraphics*[scale=1]{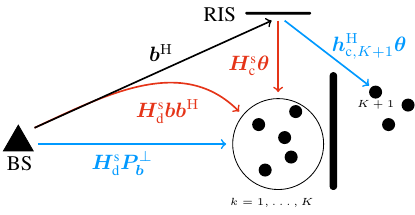}
        \caption{High-SNR scenario for ZF. In case of DPC, $\bH_\tc^\ts\bthet\bm{b}^\Her$ and $\bH_\td^\ts \bm{b}\bm{b}^\Her$ are cancelled by coding.}
        \label{fig:Szenario}
    \end{figure}
\subsection{Orthogonality of $\bm{G}$ with the strong users' direct channels}\label{subsec:OrthGHd}
We start by analyzing  $\overline{\ac{SE}}_{\text{DPC},\td}$ and $ \overline{\ac{SE}}_{\text{Lin},\td}$ which are not depending on the phases $\bthet$.
Interestingly, these are exactly the asymptotic expressions for \ac{DPC} and zero-forcing for a conventional system without the \ac{RIS} (\hspace*{-0.1cm}\cite{ZFDPCCompare}, \cite{HighSNRInstant}) and the term $\Delta{\overline{\text{SE}}}_{\td}$ in \eqref{eq:DeltadRate} is exactly the corresponding high-\ac{SNR} offset (see \cite[eq. (9)]{HighSNRInstant}).
The only difference to the conventional system in the expressions above is that instead of the Gram channel matrix $\bH_{\td}^{\text{s}}\bH_{\td}^{\text{s},\Her}$, we have the matrix $\bC_\ts=\bH_{\td}^{\text{s}}\bm{P}_{\bm{b}}^{\perp}\bH_{\td}^{\text{s},\Her}$.
Hence, the direction $\bm{b}$, corresponding to the \ac{BS}-\ac{RIS} channel, is excluded.
Excluding this direction is a major difference and deteriorates the performance.
Also, in this case, \ac{DPC} has the typical advantage over linear precoding in the sense that
\begin{equation}
    {\overline{\ac{SE}}}_{\text{DPC},\td}\ge\overline{\ac{SE}}_{\text{Lin},\td}
\end{equation}
holds with equality if $\bC_\ts$ is diagonal (see, e.g., \cite{HighSNRInstant}).
However, in comparison to the conventional case, for $\bC_\ts$ to be diagonal, not only the strong users' direct channels have to be mututally orthogonal but they also have to be orthogonal to $\bm{b}$.

For example, if $\bm{b}$ were orthogonal to the strong direct channels, i.e., $\bm{b} \in \mathrm{null}(\bH_\td^{\ts})$,  the orthogonal projector completely vanishes and we have $\bC_\ts = \bH_{\td}^{\text{s}}\bH_{\td}^{\text{s},\Her}$.
Therefore, we arrive at the conventional asymptotic expression and the negative impact of $\bm{P}_{\bm{b}}^{\perp}$ completely vanishes which is clearly a best-case scenario.

On the contrary, in case the direction $\bm{b}$ (the \ac{BS}-\ac{RIS} channel) would lie within the space spanned by the strong channels, i.e., $\bm{b} \in \mathrm{range}(\bH_\td^{\ts,\Her})$, we have a worst-case scenario in which an eigenvalue of $\bH_{\td}^{\text{s}}\bm{P}_{\bm{b}}^{\perp}\bH_{\td}^{\text{s},\Her}$ completely vanishes.
To see this, we rewrite the matrix $\bC_\ts$ as
\begin{equation}
    \bC_\ts  = \bU_\ts \bm{\Sigma}_\ts (\eye - \bV^\Her_\ts\bm{b}\bm{b}^\Her\bV_\ts)\bm{\Sigma}_\ts  \bU_\ts^\Her  
\end{equation}
where we use the \ac{SVD} $ \bH_{\td}^\ts = \bU_\ts \bm{\Sigma}_\ts \bV_\ts^\Her$.
Because $\bm{b} \in \mathrm{range}(\bH_\td^{\ts,\Her})$, we have $\bV_\ts\bV_\ts^\Her \bm{b} = \bm{b}$ and, therefore, $\bm{b}^\Her\bV_\ts\bV_\ts^\Her \bm{b}=1$.
It follows that $\norm{\bV_\ts\bm{b}}_2 =1$ and $\bm{P}^{\perp}_{ \bV_\ts\bm{b}}=\eye - \bV_\ts\bm{b}\bm{b}^\Her\bV_\ts^\Her\in \cmplx{K \times K}$ is another orthogonal projector.
However, this time, the projector $\bm{P}^{\perp}_{ \bV_\ts\bm{b}}$ is $K \times K$ instead of $\NB \times \NB$ and one singular value of the strong users is completely canceled.
A complete stream is missing in this case and the transmission is clearly deteriorated.
In a practical scenario, the orthogonality highly depends on the ratio of base station antennas $\NB$ to the number of users with non-negligible direct channel $K$.
In case $\NB \approx K +1$, the performance will be clearly degraded in comparison to $\NB \gg K+1$, where the channels are likely to be close to orthogonal.
In summary, the orthogonality of the \ac{BS}-\ac{RIS} channel with the strong users' direct channels is important for the performance in a \ac{RIS}-aided scenario and is a necessary condition for \ac{DPC} and linear precoding to perform equally.
\subsection{Channel Mitigation}\label{subsec:Inf}

We focus now on the expressions $\overline{\ac{SE}}_{\text{DPC},\tre}$ and $\overline{\ac{SE}}_{\text{Lin},\tre}$ which contain the phase manipulations at the \ac{RIS}.
Interestingly, \ac{DPC} and linear precoding have the same expressions with the major difference that the channel mitigation term $\bbthet^{\Her} \bD_\ts^{\Her}\bm{C}_\ts^\inv\bD_\ts \bbthet$ completely disappears in case of \ac{DPC}, resulting in $\Delta{\overline{\text{SE}}}_{\tre}$ in \eqref{eq:DeltarRate}.
This is an important observation, meaning that improving the weak user's channel gain is only optimal in the case of \ac{DPC}; for linear precoding, a more sophisticated solution has to be constructed which takes into account the dependence of the other users.
\subsubsection{Influence of the Direct Channels}
From the mitigation term we can see that not only the channel from the \ac{RIS} to the strong users $\bH_\tc^\ts\bthet$ but the sum together with the direct channel component $\bH_\td^\ts\bm{b}$ has a negative impact on the system performance.
This is important as this means that even when the strong users are significantly separated from the weak user, the mitigation term is still non-zero and we would experience a performance drop in $\overline{\ac{SE}}_{\text{Lin},\tre}$.
For example, if the reflecting channels within the mitigation term could be completely negelected (this refers to $\bm{H}_{\text{c}}^\ts\bthet = \bm{0}$ ) and, hence, the \ac{RIS} had only impact on the user $K+1$,
the mitigation term is still non-zero and reads as 
\begin{align}\label{eq:interferencenoreflecting}
    1 + \bbthet^{\Her} \bD_\ts^{\Her}\bm{C}_\ts^\inv\bD_\ts \bbthet &=  \frac{1}{\bm{b}^\Her \bm{P}_{{\bH_\td^{\ts,\Her}}}^{\perp}\bm{b}}
 \end{align}
 which can be shown by the matrix inversion lemma.
Hence, if the \ac{BS}-\ac{RIS} channel is not orthogonal to the direct channel, the mitigation given in \eqref{eq:interferencenoreflecting} could take any positive value and the performance is deteriorated for linear precoding.
The impact of the reflecting channels can in this regard be actually beneficial as $\bH_\tc^\ts\bthet$ could compensate to an extent for $\bH_\td^\ts\bm{b}$.
\subsubsection{Influence of the Reflecting Channels}\label{subsubsec:InflRefl}
The drawback of linear precoding will be analyzed more deeply in the following based on the dependence of the reflecting elements.
This drawback is especially pronounced in the case where $\bthet$ is chosen independently of the mitigation term in the denominator (this holds for the important cases when only the numerator, i.e., the weak user, is maximized or statistical/random phases are considered).
The focus is on this scenario where the advantage of \ac{DPC} is especially pronounced.
We can directly bound the \ac{SE} of linear precoding by dropping the mitigation term as
\begin{equation}
    \overline{\text{SE}}_{\text{Lin,\tre}} \le  {\log_2\left(\abs{\bh_{\tc,K+1} ^\Her \bthet}^2 \bar{p} \right)} = \overline{\text{SE}}_{\text{DPC},\tre}
\end{equation}
and, hence, also for this part, \ac{DPC} is always leading to a higher rate than linear precoding.
This bound is of course only tight when the mitigation term is negligible in which the two schemes perform similarly.
However, when the mitigation does play a role (which happens, e.g., in the important case where \mbox{$\NRe \rightarrow \infty$}), a tighter bound has to be derived.
In order to derive a new bound, we first establish the following lemma for the exponential integral.
\begin{lemma}\label{lemma:ExpIntegral}
    For $x>0$, the exponential integral $E_1(x)$ can be lower bounded by 
    \begin{equation}
        E_1(x)e^x > \ln\left(1+\frac{e^{-\gamma}}{x}\right), \quad \forall x >0.
    \end{equation}
\end{lemma}
\begin{proof}
    See Appendix \ref{app:ExpIntNewLowerBound}.
\end{proof}
Using Lemma \ref{lemma:ExpIntegral}, we can state the following upper bound.
\needspace{8\baselineskip}
\begin{theorem}\label{theorem:UpperboundGeneral}
    Under the  assumption of Rayleigh fading for the reflective channels $\bh_{\tre,k} \sim \gaussdist{\bm{0}}{\bm{R}_{\tre,k}} \; k\le K$,
     expression $\mathbb{E}[{\overline{\text{SE}}_{\text{Lin,\tre}}}]$ from \eqref{eq:HSNRZF} can be upper bounded as
\begin{equation}\label{eq:SELinInfBound}
    \expct{\overline{\text{SE}}_{\text{Lin,\tre}}} \le \expct{\log_2\left(\frac{\abs{\bh_{\tc,K+1}^\Her \bthet}^2\bar{p}}{   e^{-\gamma}\summe{k=1}{K}\frac{\bthet^\Her \bm{R}_{\tc,k}\bthet}{\tr(\bm{R}_{\td,k})}  }\right)},\\
\end{equation}
where $\gamma$ is the Euler-Mascheroni constant, $\bm{R}_{\tdk}$ are the covariance matrices of the direct channels, and 
\begin{equation}
    \bm{R}_{\tc,k} = \expct{\bh_{\tck}\bh_{\tck}^{\Her}} = \diag(\bm{a}^*)\bm{R}_{\tre,k}  \diag(\bm{a}) L_{\text{G}} \NB
\end{equation}
are the covariance matrices of the cascaded channels.
\end{theorem}
\begin{proof}
    See Appendix \ref{app:LinearUpperBoundErgodicRate}.
\end{proof}
It is important to note that the channel distribution of the direct channels $\bh_{\tdk}, \; \forall \, k=1,\dots,K$ as well as the channel distribution of the weak user $\bh_{\tre,K+1}$ are abitrary and \mbox{Theorem \ref{theorem:UpperboundGeneral}} holds for all distributions.
Especially, the user $K+1$ which is served by the \ac{RIS} can certainly have a \ac{LOS} component and, hence, it is important that the bound also holds when $\bh_{\tre,K+1}$ is, e.g., Rician distributed.  

In the following, we analyze Theorem \ref{theorem:UpperboundGeneral} under the assumption that all channels follow i.i.d. Rayleigh fading, i.e., $\bh_{\td/\trek} \sim \gaussdist{\bm{0}}{{L_{\td/\trek}} \eye} \; \forall k=1,\dots,K+1$ where ${L_{\td/\trek}} $ is the pathloss of user $k$.
\begin{corollary}
    For random/statistical phase shifts when assuming i.i.d. Rayleigh fading for the channels $\bh_{\td/\trek} \sim \gaussdist{\bm{0}}{{L_{\td/\trek}} \eye} \; \forall k=1,\dots,K+1$, the \ac{SE} of linear precoding is upper bounded by
    \begin{equation}\label{eq:SINRIIDRand}
        \expct{\overline{\ac{SE}}_{\text{Lin},\tre}} \le  \log_2\left(\NB  \frac{L_{\tre,K+1}   \bar{p} }{  \summe{k=1}{K} \frac{L_{\tre,k}}{L_{\td,k}} }\right)
    \end{equation}
    which is independent of $\NRe$.
    For \ac{DPC}, on the other hand, the \ac{SE} increases monotonically with $\NRe$ according to 
    \begin{equation}
        \expct{\overline{\ac{SE}}_{\text{DPC},\tre}} =  {\log_2\left( e^{-\gamma}L_G L_{\tre,K+1} \NB \NRe \bar{p} \right)}.\\
    \end{equation}
\end{corollary}
\begin{proof}
    See Appendix \ref{app:LinearDPCUpperBoundErgodicRateiidRayleighRandom}.
\end{proof}
It is important to note that both, random and statistical phase shifts, lead to the same expression for the upper bound in the case of i.i.d Rayleigh fading.
Hence, both random and statistical phase shifts have a clear limitation in the sense that the \ac{SE} is bounded by a term which does not depend on the number of reflecting elements $\NRe$.
Therefore, the mitigation of the users $k\le K$ on the user $K+1$ can be so large that there is no improvement anymore when increasing the number of reflecting elements $\NRe$.
This is clearly different from \ac{DPC} for which ${{\text{SE}}_{\text{DPC},\tre}}$ scales linearly inside the logarithm w.r.t. $\NRe$ for random/statistical phase shifts.

Additionally, when solely maximizing the channel gain of the weak user $K+1$ we obtain the following result.
\begin{corollary}
    When solely optimizing the channel gain of the weak user $K+1$, i.e., choosing $\bthet = \exp(\mathrm{j}\arg(\bh_{\tc,K+1}))$, then, under the assumption of i.i.d. Rayleigh fading for the channels $\bh_{\td/\trek} \sim \gaussdist{\bm{0}}{{L_{\td/\trek}} \eye} \; \forall k=1,\dots,K+1$, the \ac{SE} of linear precoding is upper bounded by
    \begin{equation}\label{eq:SINRIIDWeakINST}
        \expct{\overline{\ac{SE}}_{\text{Lin},\tre}} \le  \log_2\left(  \frac{\pi e^{\gamma}}{4} \NRe \NB\frac{L_{\tre,K+1}   \bar{p} }{  \summe{k=1}{K} \frac{L_{\tre,k}}{L_{\td,k}} }\right)
    \end{equation}
    whereas for \ac{DPC} the \ac{SE} is lower bounded according to 
    \begin{equation}
        \expct{\overline{\ac{SE}}_{\text{DPC},\tre}} \ge  {\log_2\left( e^{-\gamma}L_G L_{\tre,K+1} \NB \NRe^2 \bar{p} \right)}.\\
    \end{equation}
\end{corollary}
\begin{proof}
    See Appendix \ref{app:LinearUpperBoundErgodicRateiidRayleighInst} and Appendix \ref{app:DPCUpperBoundsInst}. 
\end{proof}
We see again a clear limitation for linear precoding in the sense that solely optimizing the weak user's channel gain is upper bounded by a rate expression which scales only linearly with $\NRe$ inside the logarithm.
For \ac{DPC}, this is the optimal solution and the \ac{SE} scales quadratically within the logarithm.

In summary, it is important to suppress the mitigation term, which is possible in case of instantaneous \ac{CSI} for $\NRe > K$.
Under statistical \ac{CSI}, suppressing the mitigation is only possible if the sum of the covariance matrices $\bm{R}_{\tc,k}$ of the strong users do not span the whole space $\cmplx{\NRe \times \NRe}$ which can only happen for correlated Rayleigh fading.
This analysis, however, will be subject to future work. 
In this article, only i.i.d. Rayleigh fading is considered in which linear precoding has a major disadvantage when statistical or random phase shifts are considered.

\vspace*{-0.0cm}
\section{Equivalence of DPC and Linear Precoding}
Taking into account all the results of the last section, under the assumption of orthogonal strong direct users' channels $\bH_\td^\ts$, an orthogonal \ac{BS}-\ac{RIS} channel ($\bm{b} \in \mathrm{null}(\bH_\td^{\ts})$), and that the strong users are considerable far away from the \ac{RIS} such that the reflections can be neglected,
\ac{DPC} and linear precoding perform equally w.r.t. the sum-\ac{SE}.
This can also be seen by rewriting the composite channel matrix as [cf. \eqref{eq:CompositeChannelMatrixDefinition}]
\begin{equation}
     \bH = \bH_\td + \bH_\tc \bthet  \bm{b}^\Her = \begin{bmatrix}
        \bm{\Sigma}_\ts & \bH_\tc^\ts \bthet\\
        \bm{0}&  \bh_{\tc,K+1}^\Her\bthet 
     \end{bmatrix}
     \begin{bmatrix}
         \bV_\ts^{,\Her}\\
         \bm{b}^\Her
     \end{bmatrix}
\end{equation}
where $\bU_\ts = \eye$ of the \ac{SVD} of $\bH_\td^\ts = \bU_\ts \bm{\Sigma}_\ts \bV_\ts^\Her$ as the direct channel users are orthogonal.
Because $\bH_\tc^\ts \bthet$ is negligible by assumption and $\bm{b} \perp \bV^{\ts}$, the channel matrix is orthogonal.
These are a lot of assumptions (which are more likely to be fulfilled for large $\NB$), and in general, \ac{DPC} will lead to superior performance.
\vspace*{-0.0cm}
\section{Results}
\label{sec:Results}
For the simulations we consider, similar to \cite{LOSZeroForc}, one \ac{BS} at (0, 0, 10)$\,\text{m}$ with $\NB$ antennas together with an RIS at (100, 0, 10)$\,\text{m}$ which has $\NRe$ reflecting elements.
We assume $K+1=4$ single-antenna users on a height of $1.5\,\text{m}$ which are uniformly distributed in a circle with radius $5\,\text{m}$ centered at (95, 10, 1.5)$\,\text{m}$.
We use the logarithmic pathloss model $ L_{\text{dB}} = \alpha + \beta 10\log_{10}(\frac{d}{\mathrm{m}})$ for all channels where $d$ is the distance in meter.
 The values for $\alpha$ and $\beta$ are chosen as in \cite{Relay} which are according to the 3GPP model.
For the pathloss, we choose
     $L_{\text{dB},\text{weak}} = 35.1 + 36.7\log_{10}(d/\text{m})$
for the direct channel.
For the channel between the \ac{RIS} and the users we assume a stronger channel and use 
    $L_{\text{dB},\text{strong}} = 37.51 + 22\log_{10}(d/\text{m})$.
The weak user which should be served via the \ac{RIS} has a significantly degraded direct channel (additional pathloss of $60$ dB).
For the direct and the \ac{RIS}-user channels, we assume i.i.d. Rayleigh fading whereas for the \ac{BS}-\ac{RIS} channel, we assume a \ac{LOS} channel corresponding to the outer product of two half-wavelength \ac{ULA} vectors where the \ac{AoA} and the \ac{AoD} are both given by $\frac{\pi}{2}$.
The pathloss for the \ac{LOS} \ac{BS}-\ac{RIS} channel is given by $L_{\text{dB},\text{LOS}} = 30 + 22\log_{10}(d/\text{m})$. 
In all simulations, a noise power of $\sigma^2 = -110\,\text{dBm}$ is used.

\begin{figure}[t!]
    \flushleft
    \hspace*{17pt}
    \vspace*{-16pt}
    \includegraphics*[scale=1]{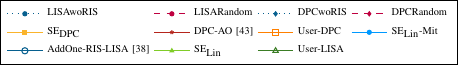}
    \flushleft

    \hspace*{-6pt}

        \centering
        \includegraphics*[scale=1]{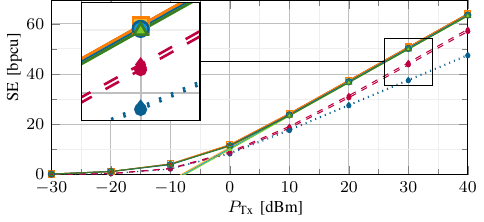}

            \label{fig:SNRPlot}

    
    \caption{ \ac{SE} evaluation with $\NB=12$ and $\NRe=64$.}
    \label{fig:SNRPlot}

\end{figure}

We evaluate the asymptotic sum-\ac{SE} expressions \eqref{eq:HSNRZF} and \eqref{eq:HSNRDPC} derived in this article where the optimal solution of \eqref{eq:HSNRDPC} w.r.t. $\bthet$ is given in closed-form by $\exp(\mathrm{j}\angle(\bh_{\tc,K+1}))$, referred to as \textbf{$\overline{\text{SE}}_{\text{DPC}}$}.
For \eqref{eq:HSNRZF}, we choose $\exp(\mathrm{j}\angle(\bh_{\tc,K+1}))$ as an intialization to the element-wise algorithm of \cite{Secrecy}, \cite{IRSLISA} that takes the mitigation term into account, referred to as \textbf{$\overline{\text{SE}}_{\text{Lin}}$-Mit}.
We compare this with simply taking the initial solution and neglecting the mitigation, referred to as \textbf{$\overline{\text{SE}}_{\text{Lin}}$}.
As a comparison, we use the \textbf{\ac{DPC}-AO} of \cite{MaxSumRateJour} with 10 random initial phase shifts of which the best solution is taken.
For linear precoding, we choose the \textbf{AddOne-RIS-LISA} algorithm of \cite{LOSZeroForc}, which converges to zero-forcing if $\NB \ge K+1$ and $P_{\text{Tx}} \rightarrow \infty$.
Additionally, we evaluate two user-based methods in which the phases of the \ac{RIS} are aligned such that the channel gain of a particular user $k$ is maximized.
This is done for all users and conventional \ac{DPC} as well as conventional linear precoding (with \ac{LISA}) is used for all the $k=1,2,\dots,K+1$ phase shifts.
Afterward, the best sum-\ac{SE} is taken and we obtain \textbf{User-DPC} and \textbf{User-LISA}, respectively.
For linear precoding, we use \ac{LISA} \cite{OriginalLISA}, which converges to zero-forcing for $\NB \ge K+1$ and $P_{\text{Tx}} \rightarrow \infty$.
Additionally, we compare all the methods with random phase shifts as well as the situation where no \ac{RIS} is present.
Both are evaluated for \ac{DPC} as well as \ac{LISA} and we obtain \textbf{LISAwoRIS}, \textbf{DPCwoRIS}, \textbf{LISARandom}, and \textbf{DPCRandom}.

In Fig. \ref{fig:SNRPlot}, we can see that all algorithms match the high-\ac{SNR} expressions after around $10\,\text{dBm}$.
Additionally, we can observe that the slope of the curves, when including the \ac{RIS}, are higher since, additionally, the weak user is allocated.
Furthermore, all optimized algorithms show similar performance as we already have $\NB=12$ \ac{BS} antennas.


\begin{figure}[t!]
    \flushleft
    \hspace*{18.25pt}
    \vspace*{-21pt}
    \includegraphics*[scale=1]{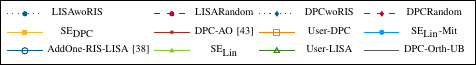}

    \flushleft

    \subfigure[Orth. BS-RIS Channel]{
        \hspace*{-8pt}

        \includegraphics*[scale=1]{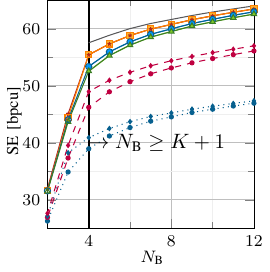}

    \label{fig:Influenceofborth}

    }\hspace*{-10pt}\subfigure[Non-Orth. BS-RIS Channel]
    {
        \includegraphics*[scale=1]{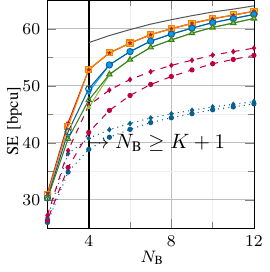}

            \label{fig:Influenceofbnonorth}

    }
    \caption{ \ac{SE} evaluation with $\NRe=64$ and $P_{\text{Tx}}=40\,\text{dBm}$.}
    \label{fig:Influenceofb}


\end{figure}

In Fig. \ref{fig:Influenceofb} we further analyze the impact of the number of \ac{BS} antennas.
Here, we additionally included \textbf{DPC-Orth-UB} which is $\text{SE}_{\text{DPC}}$ for an orthogonal \ac{BS}-\ac{RIS} channel as well as orthogonal direct user channels and, hence, serves as an upper bound.
In Fig \ref{fig:Influenceofborth}, we artificially set the vector $\bm{b}$ to be orthogonal to the strong users' direct channels whereas in Fig. \ref{fig:Influenceofbnonorth}, we don't have this artificial scenario and, therefore, the \ac{BS}-\ac{RIS} channel is not orthogonal to the direct channel.
While for $\NB \rightarrow \infty$ the algorithms converge to the same values in both plots, we can see that an orthogonal \ac{BS}-\ac{RIS} channel is clearly beneficial for the performance (especially for lower $\NB$ values), which is particularly pronounced for linear precoding.


\begin{figure}[b!]
    \vspace*{-0.4cm}
    \flushleft
    \hspace*{18.5pt}
    \vspace*{-24pt}
    \includegraphics*[scale=1]{Figures/Figureseffekt-figure_crossref2.pdf}

    \flushleft

    \subfigure[Algorithms]{
        \hspace*{-8pt}
        \includegraphics*[scale=1]{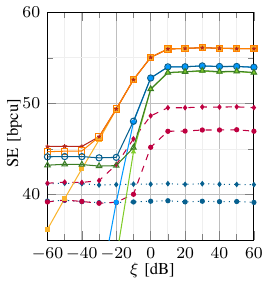}

    \label{fig:OrthogonalityReal}

    }\hspace*{-10pt}\subfigure[$\overline{\text{SE}}_{\text{DPC/Lin}}$]
    {
        \includegraphics*[scale=1]{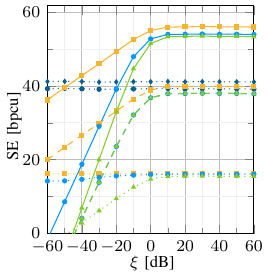}

            \label{fig:OrthogonalitySplit}

    \vspace{-5pt}
    }
    \caption{ \ac{SE} evaluation with $P_{\text{Tx}} = 40\,\text{dBm}$, $\NRe = 64$, and $\NB = 4$.}
    \label{fig:Orthogonality}

\end{figure}

In Fig. \ref{fig:Orthogonality}, we further analyze the orthogonality of the \ac{BS}-\ac{RIS} channel for $\NB=K+1=4$.
We express the orthogonality by constructing $\bm{b}^\prime = \frac{\bV_\ts \bm{1}}{\norm{\bV_\ts \bm{1}}_2} + \xi \frac{\bv^{\perp}}{\norm{\bv^{\perp}}_2} $ where $\bv^{\perp}$ is orthogonal to $\mathrm{range}(\bV_\ts)$.
Afterwards $\bm{b}$ is obtained by normalizing $\bm{b}^\prime$, i.e., $\bm{b} = \frac{\bm{b}^\prime}{\norm{\bm{b}^\prime}_2}$.
We can see in \ref{fig:OrthogonalityReal} that an orthogal \ac{BS}-\ac{RIS} channel is clearly beneficial for all methods.
The linear precdoing methods are especially sensitive for this orthogonality in comparsion to the \ac{DPC} based schemes which are more robust in this regard.

We further analyze this behavior in Fig. \ref{fig:OrthogonalitySplit} by splitting $\overline{\text{SE}}_{\text{DPC/Lin}}$(-Mit) (\textbf{solid}) into  $\overline{\text{SE}}_{\text{DPC/Lin},\td}$(-Mit) (\textbf{dashed}) and $\overline{\text{SE}}_{\text{DPC/Lin},\tre}$(-Mit) (\textbf{dotted}), see \eqref{eq:HSNRZF} and \eqref{eq:HSNRDPC}.
It is apparent that $\overline{\text{SE}}_{\text{DPC/Lin},\td}$ clearly depends on the orthogonality, especially for the linear precoding methods.
For high values of $\xi$, the orthogonal projector of $\bC_\ts = \bH_\td^\ts \bm{P}_{\bm{b}}^{\perp} \bH_\td^{\ts,\He}$ vanishes and we have the same expression as in case of only the direct channels.
The offset of $\overline{\text{SE}}_{\text{DPC/Lin},\td}$ to the direct channel for large values of $\xi$ is only due to the fact that the power is divided over $K+1=4$ users instead of $K=3$ users and, hence, we obtain the offset $K \log_2(\frac{K+1}{K}) = 1.2451$ bpcu.
On the other hand, when considering the terms $\overline{\text{SE}}_{\text{Lin/DPC},\tre}$, only the linear methods depend on $\bm{b}$.
When taking the mitigation into account, this dependence is very small, as in this scenario the \ac{RIS} has engough impact on the strong users to cancel this dependence.
However, the interfernce has to be taken actively into account as can be seen when only maximizing the weak user's channel gain.
Here, $\overline{\text{SE}}_{\text{Lin},\tre}$ significantly depends on the orthogonality as the direction $\bH_\td^\ts \bm{b}$ is not compensated when maximizing only the weak user's channel gain. 
This results in a clear degradation when the \ac{BS}-\ac{RIS} channel is not orthogonal.

\begin{figure}[t!]
    \flushleft
    \hspace*{18pt}
    \includegraphics*[scale=1]{Figures/Figureseffekt-figure_crossref2.pdf}
    \flushleft
    \vspace*{-21pt}

    \subfigure[Algorithms]{
        \hspace*{-8pt}
        \includegraphics*[scale=1]{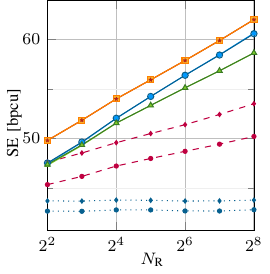}

    \label{fig:ReflectingAlgo}

    }\hspace*{-10pt}\subfigure[ $\overline{\text{SE}}_{\text{DPC}/\text{Lin},\tre}$]
    {
        \includegraphics*[scale=1]{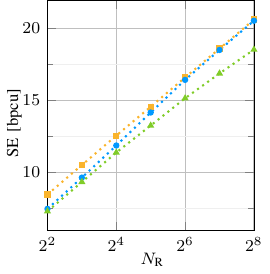}

            \label{fig:ReflectingSplit}

    }
    \vspace{-8pt}
    \caption{ \ac{SE} evaluation with $P_{\text{Tx}} = 40\,\text{dBm}$ and $\NB = 6$.}
    \label{fig:ReflectingPlot}

\end{figure}

In Fig. \ref{fig:ReflectingPlot}, we further investigate the ability to compensate for the channel part $\bH_\td^\ts \bm{b}$ depending on the number of reflecting elements $\NRe$.
At first, we focus on Fig. \ref{fig:ReflectingAlgo} where we can see that there is a clear gap between \ac{DPC} and linear precoding.
For random phase shifts as well as when neglecting the mitigation term, this gap is significantly larger as in case of having solely the direct channels.
This is differnet for $\overline{\text{SE}}_{\text{Lin}}$-Mit which for a low number of reflecting elements performs equal to $\overline{\text{SE}}_{\text{Lin}}$, whereas for a high number of elements, the gap to \ac{DPC} is significantly reduced.

We analyze this behavior in \ref{fig:ReflectingSplit} where we only plot $\overline{\text{SE}}_{\text{DPC}/\text{Lin},\tre}$. 
Note that $\overline{\text{SE}}_{\text{DPC}/\text{Lin},\td}$ does not depend on the reflecting elements and, hence, is constant over $\NRe$.
We can see that if the impact of the reflecting channel is low, we have $ 1 + \bbthet^{\Her} \bD_\ts^{\Her}\bm{C}_\ts^\inv\bD_\ts \bbthet =  1/\bm{b}^\Her \bm{P}_{{\bH_\ts^\Her}}^{\perp}\bm{b}$ according to eq. \eqref{eq:interferencenoreflecting} and maximizing the weak user's channel gain is optimal.
Hence, $\overline{\text{SE}}_{\text{Lin},\tre}$ and  $\overline{\text{SE}}_{\text{Lin},\tre}$-Mit are overlapping and a clear gap to \ac{DPC} exists.
For a higher number of reflecting elements the reflecting channel has more impact and actually helps to compensate for $\bH_\td^\ts \bm{b}$.
Therefore, $\overline{\text{SE}}_{\text{Lin},\tre}$-Mit gains a lot of performance over $\NRe$ and is very close to  $\overline{\text{SE}}_{\text{DPC},\tre}$.
On the other hand, the increasing impact of the reflecting channel becomes problematic when the mitigation is neglegted and we can already see a saturation effect of  $\overline{\text{SE}}_{\text{Lin},\tre}$ for a higher number of reflecting elements.

In Fig. \ref{fig:Bound}, this saturation effect is investigated.
To clearly see the problem of the mitigation term, we assume the direct channels to have an additional pathloss of $20\,\text{dB}$.
In Fig \ref{fig:BoundAlgo} we can see that the mitigation term plays a fundamental role for the system performance.
For random  and statistical phase shifts, the \ac{SE} stays constant whereas solely maximizing the weak user's channel gain scales only linearly w.r.t. $\NRe$ inside the logarithm.
This is further illustrated by plotting the bounds in \eqref{eq:SINRIIDRand} and \eqref{eq:SINRIIDWeakINST} which explain these effects.
In Fig. \ref{fig:BoundSplit} this is further supported by recognizing that $\overline{\text{SE}}_{\text{Lin},\tre}$ is clearly limited when the mitigation is not taken into account.
On the other hand, when the mitigation is taken into account, we can expect the same scaling behavior as \ac{DPC}.

\begin{figure}[t!]
    \flushleft
    \hspace*{18pt}
    \vspace*{-21pt}
    \includegraphics*[scale=1]{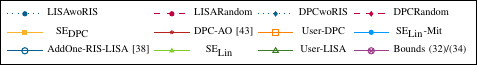}
    \flushleft

    \subfigure[Algorithms/Bounds]{
        \hspace*{-8pt}

        \includegraphics*[scale=1]{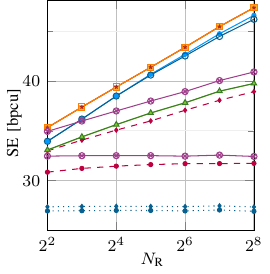}

    \label{fig:BoundAlgo}

    }\hspace*{-10pt}\subfigure[$\overline{\text{SE}}_{\text{DPC/Lin}}$]
    {
        \includegraphics*[scale=1]{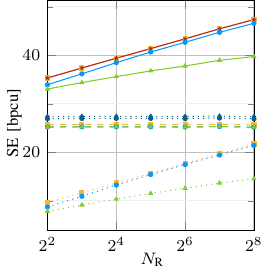}

            \label{fig:BoundSplit}

    \vspace{-5pt}
    }
    \caption{ \ac{SE} evaluation with $P_{\text{Tx}} = 40\,\text{dBm}$ and $\NB = 12$. Additional $20\,\text{dB}$ pathloss for the direct channels.}
    \label{fig:Bound}

\end{figure}

\vspace*{-0.3cm}
\section{Conclusion}
We have analyzed a scenario in which a group of strong users and a group of weak users are supported by a \ac{RIS} with a \ac{LOS} dominated \ac{BS}-\ac{RIS} channel.
Interestingly, linear precoding differs significantly from \ac{DPC}, i.e., that it suffers from an mitigation term which completely vanishes in case of \ac{DPC}.
This is especially pronounced when considering random/statistical phase shifts under i.i.d. Rayleigh fading.
A number of assumptions have to be fulfilled for \ac{DPC} and linear precoding to perform equally, e.g., the \ac{BS}-\ac{RIS} channel has to be orthogonal to the users with non-negligible direct channels which is generally important for the performance.
In future works, it is analyzed how these results can be extended to a scenario with an increased rank of the \ac{BS}-\ac{RIS} channel.
\vspace*{-0.2cm}

\appendix
\subsection{Effective Channel Gains of ZF}\label{app:ZFEffectiveGains}
We start by rewriting the Gram channel matrix 
\begin{align}
    \bm{H}_{}\bm{H}_{}^{\He} &= \begin{bmatrix}
        \bm{C}_\ts & \bm{0}\\
        \bm{0}^\transpo& 0
    \end{bmatrix} + \bm{D} \bbthet \bbthet^\Her \bm{D}^\Her.
\end{align}
From the definitions \eqref{eq:DiDefinition} and \eqref{eq:DStrongDef} it follows that
\begin{equation}
    \bD\bbthet = \begin{bmatrix}
        \bD_\ts\bbthet\\
        \bh_{\text{c},K+1}^\Her \bthet
    \end{bmatrix}
\end{equation}
and we can express the Gram channel matrix as
\begin{equation}
    \bH\bH^\Her =     \begin{bmatrix}
        \bm{C}_\ts+\bm{D}_\ts \bbthet \bbthet^\Her \bm{D}_\ts^\Her & \bm{D}_\ts \bbthet\bthet^\Her  \bh_{\text{c},K+1}\\
       \bh_{\text{c},K+1}^\Her\bthet\bbthet^\Her \bm{D}_\ts^\Her&  \abs{\bh_{\text{c},K+1}^\Her\bthet}^2
    \end{bmatrix} .
\end{equation}
Applying the block matrix inversion, we have 
\begin{align}
  \nonumber
   & (\bH\bH^\Her)^{\inv}=\begin{bmatrix}
        \bm{C}_\ts+\bm{D}_\ts \bbthet \bbthet^\Her \bm{D}_\ts^\Her & \bm{D}_\ts \bbthet   \bthet^\Her  \bh_{\text{c},K+1}\\
       \bh_{\text{c},K+1}^\Her\bthet\bbthet^\Her \bm{D}_\ts^\Her&  \abs{\bh_{\text{c},K+1}^\Her\bthet}^2
    \end{bmatrix}^{\inv}\\
    &= 
    \begin{bmatrix}
        \bm{C}_\ts^{\inv} & -\bm{C}_\ts^{\inv}\bm{D}_\ts \bbthet \frac{ \bthet^\Her  \bh_{\text{c},K+1}}{ \abs{\bh_{\text{c},K+1}^\Her\bthet}^2}\\
       -\frac{\bh_{\text{c},K+1}^\Her\bthet}{\abs{\bh_{\text{c},K+1}^\Her\bthet}^2}\bbthet^\Her \bm{D}_\ts^\Her\bm{C}_\ts^{\inv}& \frac{1+\bbthet^{\Her} \bD_\ts^{\Her}\bm{C}_\ts^{\inv} \bD_\ts \bbthet}{\abs{\bh_{\tc,K+1}^\Her \bthet}^2}
    \end{bmatrix}
\end{align}
and, hence, we can give the inverted channel gains for zero forcing $\bm{e}_k^\transpo (\bH\bH^\Her)^{\inv} \bm{e}_k$ according to \eqref{eq:ZFEffectiveGains}. 
\subsection{Asymptotic SE of DPC}\label{app:DPCAsymSE}
Using \eqref{eq:HHDttDDef} and defining 
\begin{equation}
    \bC = \begin{bmatrix}
        \bm{C}_\ts & \bm{0}\\
        \bm{0}^\transpo& 0
    \end{bmatrix},
\end{equation} we can write the \ac{SE} in case of \ac{DPC} as
\begin{align}
    \nonumber  &\text{SE}_{\text{DPC}} = \log_2\det\left(\eye + \bH^\Her\bQ \bH \right)\\
        \nonumber   &=\log_2\det\left(\eye + \bH\bH^\Her \bar{p}\right)\\
        \nonumber   &=  \log_2\det\left(\eye + \bC \bar{p}+\bD \bbthet \bbthet^\Her \bD^\Her\bar{p}\right)\\
         \nonumber   &=  \log_2\det\left(\eye + \bC\bar{p}\right)\\
        \nonumber   &\quad +  \log_2\det\left(\eye + \left(\eye + \bC \bar{p}\right)^{\inv}\bD \bbthet\bbthet^\Her \bD^\Her \bar{p}\right)\\
        \nonumber   &=  \log_2\det\left(\eye + \bC_\ts\bar{p}\right)\\
        \nonumber    &\quad +  \log_2\left(1 +\bbthet^\Her \bD^\Her \left(\eye + \bC \bar{p}\right)^{\inv}\bD \bbthet   \bar{p}\right)\\
                \nonumber  &=  \summe{k=1}{K}\log_2\left(1 + \lambda_k \bar{p}\right) \\
        \nonumber  &\quad+  \log_2\left(1 + \summe{k=1}{K+1}\abs{\bbthet^\Her \bD^\Her \bu_k}^2\hspace*{-4pt}\frac{ \bar{p}}{1+\lambda_k  \bar{p}} \right)\hspace*{-3pt}\\
        \nonumber &=  \summe{k=1}{K}\log_2\left(1 + \lambda_k \bar{p}\right) \\
      &\quad+  \log_2\left(1 +\abs{\bh_{\tc,K+1}^\Her \bthet}^2 \bar{p}+ \summe{k=1}{K}\abs{\bbthet^\Her \bD^\Her \bu_k}^2\hspace*{-4pt}\frac{ \bar{p}}{1+\lambda_k  \bar{p}} \right)\hspace*{-3pt}
    \end{align}
where $\lambda_k$ denotes the $k$-th eigenvalue of $\bC$ and \mbox{$\abs{\bh_{\tc,K+1}^\Her \bthet}^2 \neq 0$}.

\subsection{A Lower Bound for the Exponential Integral}\label{app:ExpIntNewLowerBound}
There are popular bounds for the exponential integral (see \cite[p. 229, eq. (5.1.19) and (5.1.20)]{Abramo} and \cite{ExpInt}, \cite{ExpIntThree}), however, all of them are not suited for our purposes in this article.
Specifically, the problem with the bound of \cite{ExpInt}
\begin{equation}
    E_1(x)e^x > \frac{1}{2}\ln\left(1+\frac{2}{x}\right), \quad \forall x >0
\end{equation}
is the prefactor of $0.5$ outside the logarithm.
In \cite[eq. (12)]{ExpIntTwo} there has been given a bound without this factor 
\begin{equation}\label{eq:ExpIntConstNotApplicable}
    E_1(x)e^x > -\gamma + \ln\left(1+\frac{1}{x}\right), \quad \forall x >0.
\end{equation}
where $\gamma$ is the Euler-Mascheroni constant.
However, also this bound is not applicable for our problem because of the constant term $-\gamma$. 
Hence, we introduce a new lower bound for the exponential integral which is tighter than the one in \eqref{eq:ExpIntConstNotApplicable} for all $x >0$.
In particular, we prove the relationship
\begin{equation}\label{eq:ExpIntGoal}
    E_1(x)e^x > \ln\left(1+\frac{e^{-\gamma}}{x}\right), \quad \forall x >0
\end{equation}
and we have 
\begin{equation}
    \ln\left(e^{-\gamma}+\frac{e^{-\gamma}}{x}\right) < \ln\left(1+\frac{e^{-\gamma}}{x}\right) , \quad \forall x >0
\end{equation}
when comparing the bound with \eqref{eq:ExpIntConstNotApplicable}.

To show \eqref{eq:ExpIntGoal}, we introduce the difference 
\begin{equation}\label{eq:DiffEq}
   g(x) =  E_1(x) - e^{-x}\ln\left(1+\frac{e^{-\gamma}}{x}\right)
\end{equation}
which we show to be positive for all positive $x$.
By exploiting the property $E_1^\prime(x) = - E_0^\prime(x) = -\frac{e^{-x}}{x}$  (see \cite[p. 230, eq. (5.1.26) and p. 229 eq. (5.1.24)]{Abramo}), the first-order derivative $g^\prime(x)$ of $g(x)$ is given by 
\begin{align*}
    g^\prime(x) = e^{-x}\underbrace{\left(\ln\left(1+\frac{e^{-\gamma}}{x}\right) - \frac{1}{e^{-\gamma} + x}\right)}_{\tilde{g}^\prime(x)}.
\end{align*}
Taking the derivative of $\tilde{g}^\prime(x)$, we obtain 
\begin{equation}
    \frac{\mathrm{d}}{\mathrm{d}x} \tilde{g}^\prime(x)  = \frac{-e^{-2\gamma} + x(1-e^{-\gamma})}{x(x+e^{-\gamma})^2}.
\end{equation}
This expression is equal to zero for $x = \frac{e^{-2\gamma}}{1-e^{-\gamma}}\approx 0.719$, negative for $x < \frac{e^{-2\gamma}}{1-e^{-\gamma}}$ and positive for $x > \frac{e^{-2\gamma}}{1-e^{-\gamma}}$.
It follows, that $ \tilde{g}^\prime(x)$ is monotonically decreasing for $x < \frac{e^{-2\gamma}}{1-e^{-\gamma}}$, having a minimum at $x = \frac{e^{-2\gamma}}{1-e^{-\gamma}}$ and then monotonically increasing for $x > \frac{e^{-2\gamma}}{1-e^{-\gamma}}$.
Additionally, we have 
\begin{align*}
    \underset{x \rightarrow 0^+}{\lim} \tilde{g}^\prime(x) = \infty, \quad \underset{x \rightarrow \infty}{\lim} \tilde{g}^\prime(x) = 0.
\end{align*}
Summarizing the information about $\tilde{g}^\prime(x)$, we can infer that $\tilde{g}^\prime(x)$ has exactly one zero in the interval $x_0 \in \left(0,\frac{e^{-2\gamma}}{1-e^{-\gamma}}\right)$.
For $x<x_0$, $\tilde{g}^\prime(x)$ is positive, whereas for $x>x_0$, $\tilde{g}^\prime(x)$ is negative.
Because $g^\prime(x) = e^{-x} \tilde{g}^\prime(x)$,  ${g}^\prime(x)$ also has exactly one zero in the interval $x_0 \in \left(0,\frac{e^{-2\gamma}}{1-e^{-\gamma}}\right)$, is positive for $x<x_0$, and negative for $x>x_0$.
Hence, the difference $g(x)$ has exactly one maximum at $x_0 \in \left(0,\frac{e^{-2\gamma}}{1-e^{-\gamma}}\right)$, is strictly increasing for $x<x_0$, and strictly decreasing for $x>x_0$.
It follows that we only have to show that $\underset{x \rightarrow 0^+}{\lim} g(x) \ge 0$ and $\underset{x \rightarrow \infty}{\lim} g(x) \ge 0$.
For $x\rightarrow \infty$, one can directly see from \eqref{eq:DiffEq} that
\begin{align*}
    \underset{x \rightarrow \infty}{\lim} g(x) = 0.
\end{align*}
For $x\rightarrow 0^+$ the situation is less clear. We use the series representation of $E_1(x)$ from \cite[p.229, eq. (5.1.11)]{Abramo} and write $g(x)$ as 
\begin{align*}
    g(x) = -\gamma - \ln(x) - \summe{k=1}{\infty} \frac{(-x)^k}{k k!} - e^{-x} \ln\left(1 + \frac{e^{-\gamma}}{x}\right).
\end{align*}
Because
\begin{equation}
    \underset{x\rightarrow 0^+}{\lim}\left[ - \ln(x) - e^{-x} \ln\left(1 + \frac{e^{-\gamma}}{x}\right)\right] =  \gamma
\end{equation}
and furthermore, with the geometric series, we have
\begin{align}
   \nonumber\underset{x\rightarrow 0^+}{\lim} \summe{k=1}{\infty} \frac{(-x)^k}{k k!} &\le \underset{x\rightarrow 0^+}{\lim} \left[-1+ \summe{k=0}{\infty} x^k\right]\\
   & = \underset{x\rightarrow 0^+}{\lim}\left[ \frac{1}{1-x} - 1\right] = 0,
\end{align}
as well as 
\begin{align}
   \nonumber \underset{x\rightarrow 0^+}{\lim} \summe{k=1}{\infty} \frac{(-x)^k}{k k!} &\ge \underset{x\rightarrow 0^+}{\lim} \left[1- \summe{k=0}{\infty} x^k\right]\\
    & = \underset{x\rightarrow 0^+}{\lim}\left[1- \frac{1}{1-x} \right] = 0.
 \end{align}
It follows that 
\begin{align*}
    \underset{x \rightarrow 0^+}{\lim} g(x) = 0.
\end{align*}
Hence, we have
\begin{equation}
    g(x)> 0, \quad \forall x > 0
\end{equation} 
and the upper bound in \eqref{eq:ExpIntGoal} holds.

\subsection{Ergodic Upper Bound for Linear Precoding}\label{app:LinearUpperBoundErgodicRate}
In this section, we derive an upper bound for $\mathbb{E}[{\overline{\text{SE}}_{\text{Lin},\tre}}]$ defined in \eqref{eq:HSNRZF}.
This upper bound holds for random phase shifts, statistical phase shifts as well as for phase shifts which are only based on optimizing the numerator $|\bh_{\tc,K+1}^\Her \bthet|^2$.
All these choices have in common that they are independent of the channels of user $k=1,\dots,K$.
Firstly, we derive a lower bound for $\bbthet^\Her\bD_\ts^\Her \bC_\ts^{\inv} \bD_\ts\bbthet$ with the matrix inversion lemma by
        \begin{align}\label{eq:DenomLowerInit}
        \nonumber&\bbthet^\Her\bD_\ts^\Her \bC_\ts^{\inv} \bD_\ts\bbthet =\\
        \nonumber &=\bbthet^\Her\bD_\ts^\Her {\left((\bH_\td^\ts\bH_\td^{\ts,\Her}) - \bH_\td^{\ts} \bm{b}\bm{b}^\Her \bH_\td^{\ts,\Her}\right)}^{\inv} \bD_\ts\bbthet\\
        \nonumber&=  \bbthet^\Her\bD_\ts^\Her{(\bH_\td^\ts\bH_\td^{\ts,\Her})}^{\inv}\bD_\ts\bbthet + \frac{\abs{\bm{b}^\Her \bH_\td^{\ts,\Her} {(\bH_\td^\ts\bH_\td^{\ts,\Her})}^{\inv} \bD_\ts\bbthet}^2}{1-\bm{b}^\Her \bH_\td^{\ts,\Her}{(\bH_\td^\ts\bH_\td^{\ts,\Her})}^{\inv}\bH_\td^\ts \bm{b} }\\
        &=  \bbthet^\Her\bD_\ts^\Her{(\bH_\td^\ts\bH_\td^{\ts,\Her})}^{\inv}\bD_\ts\bbthet + \frac{\abs{\bm{b}^\Her \bH_\td^{\ts,+} \bD_\ts \bbthet}^2}{1-\bm{b}^\Her \bm{P}_{\bH_\td^{\ts,\Her}}\bm{b} }
\end{align}
where we used the definition of the matrix $\bm{C}_\ts  = \bm{H}_{\text{d}}^\ts\bm{P}_{\bm{b}}^{\perp} \bm{H}_{\text{d}}^{\ts,\Her}$, the pseudoinverse \mbox{$\bH_\td^{\ts,+}=\bH_\td^{\ts,\Her}{(\bH_\td^\ts\bH_\td^{\ts,\Her})}^{\inv}$}, and the projector \mbox{$\bm{P}_{\bH_\td^{\ts,\Her}} = \bH_\td^{\ts,\Her}{(\bH_\td^\ts\bH_\td^{\ts,\Her})}^{\inv}\bH_\td^\ts=\bH_\td^{\ts,+}\bH_\td^\ts$}.
With the definition of the matrix $\bD_\ts =[\bm{H}^\ts_{\text{c}},\;\bm{H}^\ts_{\text{d}}\bm{b}]$ we can express the first term of \eqref{eq:DenomLowerInit} as 
        \begin{align}\label{eq:DenomLowerFirst}
        \nonumber&  \bbthet^\Her\bD_\ts^\Her{(\bH_\td^\ts\bH_\td^{\ts,\Her})}^{\inv}\bD_\ts\bbthet
        \nonumber=  \bthet^\Her\bH_\tc^{\ts,\Her}{(\bH_\td^\ts\bH_\td^{\ts,\Her})}^{\inv}\bH_\tc^\ts\bthet\\
        & + 2\mathrm{Re} ( \bm{b}^\Her \bH_\td^{\ts,+} \bH_\tc^\ts \bthet) 
          + \bm{b}^\Her \bm{P}_{\bH_\td^{\ts,\Her}}\bm{b}
\end{align}
and the second term of \eqref{eq:DenomLowerInit} as
        \begin{align}\label{eq:DenomLowerSecond}
            \nonumber&\frac{\abs{\bm{b}^\Her \bH_\td^{\ts,+} \bD_\ts \bbthet}^2}{1-\bm{b}^\Her \bm{P}_{\bH_\td^{\ts,\Her}}\bm{b} }=
 \frac{\abs{\bm{b}^\Her \bH_\td^{\ts,+} \bH_\tc^\ts \bthet}^2}{1-\bm{b}^\Her \bm{P}_{\bH_\td^{\ts,\Her}}\bm{b}} \\
        &+ \frac{ 2\mathrm{Re} (\bm{b}^\Her \bH_\td^{\ts,+} \bH_\tc^\ts \bthet)\bm{b}^\Her \bm{P}_{\bH_\td^{\ts,\Her}}\bm{b}  }{1-\bm{b}^\Her \bm{P}_{\bH_\td^{\ts,\Her}}\bm{b} } + \frac{(\bm{b}^\Her \bm{P}_{\bH_\td^{\ts,\Her}}\bm{b})^2 }{1-\bm{b}^\Her \bm{P}_{\bH_\td^{\ts,\Her}}\bm{b} }.
\end{align}
Finding a common denominator for the terms \eqref{eq:DenomLowerFirst} and \eqref{eq:DenomLowerSecond} and plugging them into \eqref{eq:DenomLowerInit}
results in 
        \begin{align}\label{eq:DenomLower}
            \nonumber&\bbthet^\Her\bD_\ts^\Her \bC_\ts^{\inv} \bD_\ts\bbthet = \\
              \nonumber&=  \bthet^\Her\bH_\tc^{\ts,\Her}{(\bH_\td^\ts\bH_\td^{\ts,\Her})}^{\inv}\bH_\tc^\ts\bthet  \\
        \nonumber& +\hspace*{-2pt}\frac{2\mathrm{Re} ( \bm{b}^\Her \bH_\td^{\ts,+} \bH_\tc^\ts \bthet)\hspace*{-1pt}+\hspace*{-1pt} \bm{b}^\Her \bm{P}_{\bH_\td^{\ts,\Her}}\bm{b} \hspace*{-1pt}+\hspace*{-1pt} \abs{\bm{b}^\Her \bH_\td^{\ts,+} \bH_\tc^\ts \bthet}^2 }{{1-\bm{b}^\Her \bm{P}_{\bH_\td^{\ts,\Her}}\bm{b} }}\\
        \nonumber &=  \bthet^\Her\bH_\tc^{\ts,\Her}{(\bH_\td^\ts\bH_\td^{\ts,\Her})}^{\inv}\bH_\tc^\ts\bthet  \\
        \nonumber& +\frac{\abs{\bm{b}^\Her \bH_\td^{\ts,+} \bH_\tc^\ts \bthet+1}^2 -1 +\bm{b}^\Her \bm{P}_{\bH_\td^{\ts,\Her}}\bm{b} }{{1-\bm{b}^\Her \bm{P}_{\bH_\td^{\ts,\Her}}\bm{b} }}\\
         & \ge \bthet^\Her\bH_\tc^{\ts,\Her}{(\bH_\td^\ts\bH_\td^{\ts,\Her})}^{\inv}\bH_\tc^\ts\bthet -1
\end{align}
where we used that $1-\bm{b}^\Her \bm{P}_{\bH_\td^{\ts,\Her}}\bm{b} >0$.
With \eqref{eq:DenomLower}, we can upper bound the ergodic rate as
\begin{align}
   \nonumber &\expct{\overline{\ac{SE}}_{\text{Lin},\tre}}=  \expct{\log_2\left(\frac{\abs{\bh_{\tc,K+1}^\Her \bthet}^2\bar{p}}{1+  \bbthet^\Her\bD_\ts^\Her \bC_\ts^{\inv} \bD_\ts\bbthet }\right)}\\
   \nonumber &\le  \expct{\log_2\left(\frac{\abs{\bh_{\tc,K+1}^\Her \bthet}^2\bar{p}}{ \bthet^{\Her}\bH_\tc^{\ts,\Her}   {(\bH_\td^\ts \bH_\td^{\ts,\Her})}^\inv \bH_\tc^\ts\bthet  }\right)}\\
    \nonumber&=\expct{\log_2\left(\abs{\bh_{\tc,K+1}^\Her \bthet}^2\bar{p} \right)- \log_2 \left(\bh^\Her   {(\bH_\td^\ts \bH_\td^{\ts,\Her})}^\inv \bh  \right)}\\
\end{align}
where we introduced the definition
\begin{equation}\label{eq:hDefinition}
    \bh = \bH_\tc^{\ts} \bthet.
\end{equation}
The mitiagtion term 
\begin{align}\label{eq:Interference}
    \expct{\text{Mit}} = \underset{\bH_\tre,\bthet}{\mathbb{E}} \left[ \underset{\bH_\td^\ts|\bH_\tre,\bthet}{\mathbb{E}}\left[\log_2 \left( \bh^\Her   {(\bH_\td^\ts \bH_\td^{\ts,\Her})}^\inv \bh  \right)\right]\right]
\end{align}
is analyzed in the following.
For this, we define the matrix
\begin{equation}
    \bm{H}_{\text{d}}^{\ts,\Her} =[\bh_{\td,1},\dots,\bh_{\td,K}] = [\bm{\tilde{h}}_{\td,1},\dots,\bm{\tilde{h}}_{\td,K}] \bm{\Sigma}_\td =  \bm{\tilde{H}}_{\text{d}}^{\ts,\Her} \bm{\Sigma}_\td
\end{equation}
where $\bm{\Sigma}_\td$ is a diagonal matrix with $[\bm{\Sigma}_{\td}]_{k,k} = \sqrt{\tr(\bm{R}_{\text{\td,k}})}$ and $ \bm{\tilde{H}}_{\text{d}}^{\ts,\Her}= [\bm{\tilde{h}}_{\td,1},\dots,\bm{\tilde{h}}_{\td,K}]$ are the normalized direct channels with unit-variance.
Further defining $\bm{\tilde{h}} = \bm{\Sigma}_\td^{\inv} \bh$, we can express the quadratic form within the logarithm in \eqref{eq:Interference} as 
\begin{equation}
    \bh^\Her   {(\bH_\td^\ts \bH_\td^{\ts,\Her})}^\inv \bh  = \bm{\tilde{h}}^\Her   {( \bm{\tilde{H}}_{\text{d}}^{\ts}  \bm{\tilde{H}}_{\text{d}}^{\ts,\Her})}^\inv \bm{\tilde{h}}.
\end{equation}
Introducing the \ac{EVD} of $ {( \bm{\tilde{H}}_{\text{d}}^{\ts}  \bm{\tilde{H}}_{\text{d}}^{\ts,\Her})}^\inv  = \summe{k=1}{K} \bu_k \bu_k^\Her \frac{1}{\lambda_k}$ and applying the weighted harmonic-arithmetic mean inequality, we obtain
\begin{equation}
    \begin{aligned}
    \bm{\tilde{h}}^\Her   {( \bm{\tilde{H}}_{\text{d}}^{\ts}  \bm{\tilde{H}}_{\text{d}}^{\ts,\Her})}^\inv \bm{\tilde{h}} &= \summe{k=1}{K} \abs{\bm{\tilde{h}}^\Her   \bu_k}^2 \frac{1 }{\lambda_k} \ge  \frac{\left(\summe{k=1}{K} \abs{\bm{\tilde{h}}^\Her   \bu_k}^2 \right)^2}{\summe{k=1}{K} \abs{\bm{\tilde{h}}^\Her   \bu_k}^2 {\lambda_k}}\\
    &= \frac{\norm{\bm{\tilde{h}}}^4}{\bm{\tilde{h}}^\Her   {( \bm{\tilde{H}}_{\text{d}}^{\ts}  \bm{\tilde{H}}_{\text{d}}^{\ts,\Her})} \bm{\tilde{h}}}.
    \end{aligned}
\end{equation}
Using this result and recognizing that $\log_2(\frac{a}{x})$ is convex w.r.t. $x>0$ for any positive $a$, we can bound the inner expectation in \eqref{eq:Interference} by applying Jensen's inequality as
\begin{align}\label{eq:Upperboundwithh}
       \nonumber &\underset{\bH_\td^\ts|\bH_\tre,\bthet}{\mathbb{E}}\left[\log_2 \left( \bh^\Her   {(\bH_\td^\ts \bH_\td^{\ts,\Her})}^\inv \bh  \right)\right]\\
       \nonumber  &\ge\underset{\bH_\td^\ts|\bH_\tre,\bthet}{\mathbb{E}}\left[\log_2 \left( \frac{\norm{\bm{\tilde{h}}}^4}{\bm{\tilde{h}}^\Her   {( \bm{\tilde{H}}_{\text{d}}^{\ts}  \bm{\tilde{H}}_{\text{d}}^{\ts,\Her})} \bm{\tilde{h}}}  \right)\right]\\
       \nonumber  &\ge\log_2 \left( \frac{\norm{\bm{\tilde{h}}}^4}{{\bm{\tilde{h}}}^\Her  \underset{\bH_\td^\ts|\bH_\tre,\bthet}{\mathbb{E}}\left[ {( \bm{\tilde{H}}_{\text{d}}^{\ts}  \bm{\tilde{H}}_{\text{d}}^{\ts,\Her})} \right]\bm{\tilde{h}}}  \right)\\
       \nonumber  &=\log_2 \left( \frac{\norm{\bm{\tilde{h}}}^4}{{\bm{\tilde{h}}}^\Her  \,\eye\, \bm{\tilde{h}}}  \right)\\
       \nonumber  &=\log_2 \left(\norm{\bm{\tilde{h}}}^2 \right)\\
         &=\log_2 \left( \summe{k=1}{K}\frac{\abs{h_k}^2 }{\tr(\bm{R}_{\td,k})} \right).
\end{align}
As $\bH_\tre^\ts$ is assumed to be independent of $\bH_\td^\ts$ and it is additionally assumed that the phases are chosen independent of $\bH_\tre^\ts$ (either random, based on the statistics or only dependent of $\bh_{\tre,K+1}$),
we have $\mathbb{E}[\bH_\td^\ts \bH_\td^{\ts,\Her}|\bH_\tre,\bthet] = \mathbb{E}[\bH_\td^\ts \bH_\td^{\ts,\Her}]$.
As the users are pairwise independent and the variances are normalized to unit-norm, $\expct{\bH_\td^\ts \bH_\td^{\ts,\Her}} = \eye $ holds.
Analyzing \eqref{eq:Upperboundwithh}, we observe with \eqref{eq:hDefinition} that given $\bthet$, the variable $h_k = \bh^\Her_{\tc,k}\bthet $ is a sum of Gaussian random variables.
This holds because we have $\bh_{\trek} \sim \gaussdist{\bm{0}}{\bm{R}_{\trek}}$ and from $\bh_{\tck}^\Her = \sqrt{\NB L_{\text{G}}} \bh_{\trek}^\Her \diag(\bm{a})$, it directly follows that also $\bh_{\tck} \sim \gaussdist{\bm{0}}{\bm{R}_{\tck}}$ with 
\begin{equation}
    \bm{R}_{\tc,k} = \expct{\bh_{\tck}\bh_{\tck}^{\Her}} = \diag(\bm{a^*})\bm{R}_{\tre,k}\diag(\bm{a})  L_{\text{G}} \NB.
\end{equation}
Note that we assume $\bthet$ to be independent of $\bh_{\tre,k}$ for $k \le K$ and, hence, $h_k$ is again Gaussian distributed with mean and variance
\begin{equation}
    \expct{h_k} = 0, \quad   \text{var}[{h_k}] = \expct{\abs{\bh_{\tc,k}^\Her\bthet}^2} = \bthet^\Her \bm{R}_{\tc,k} \bthet.
\end{equation}
Therefore, defining
\begin{align}
    \xi_k = \frac{2}{ \bthet^\Her \bm{R}_{\tc,k} \bthet}   \abs{h_k}^2 \sim \chi^2(2),
\end{align}
 $\xi_k$ is chi-squared distributed with two degrees of freedom.
Combining these observations with the lower bound in \eqref{eq:Upperboundwithh}, we can lower bound the expression in \eqref{eq:Interference} by 
\begin{align}\label{eq:LowerBoundXi}
    \expct{\text{Mit}} \ge \underset{\bthet}{\mathbb{E}} \left[ \underset{\xi_{1:K}}{\mathbb{E}}\left[\log_2 \left( \summe{k=1}{K}\alpha_k \xi_k   \right)\right]\right]
\end{align}
where \mbox{$\alpha_k = \bthet^\Her \bm{R}_{\tc,k} \bthet/\big(2\tr(\bm{R}_{\td,k})\big)$}, \mbox{$\xi_{1:K} = \xi_1,\xi_2,\dots,\xi_{K}$}.
The lower bound in \eqref{eq:LowerBoundXi} will be further bounded by successively evaluating the expression based on the chain rule of probabilities where we first assume $K\ge2$.
Starting with $\xi_1$, the inner expectation in \eqref{eq:LowerBoundXi} can be written as
\begin{align}\label{eq:LawOfTotalFirststep}
    \nonumber &\underset{\xi_{1:K}}{\mathbb{E}}\left[\log_2\left( \summe{k=1}{K}\alpha_k \xi_k  \right) \right] \\
    &=\underset{\xi_{2:K}}{\mathbb{E}}\left[\underset{\xi_1}{\mathbb{E}}\left[\log_2\left( b_1 + \alpha_1 \xi_1  \right)\Big|\xi_{2:K}\right]\right].
\end{align}
where we additionally introduced the notation 
\begin{equation}
    b_1 =  \summe{k=2}{K}\alpha_k \xi_k.
\end{equation}
Using the definition of the pdf of a chi-squared random variable with two degrees of freedom $f_{\xi_1}(\xi_1) = \frac{1}{2}e^{-\frac{\xi_1}{2}}$, the inner expectation of \eqref{eq:LawOfTotalFirststep} can be expressed as 
\begin{align}
    \nonumber &\underset{\xi_1}{\mathbb{E}}\left[\log_2\left(b_1 + \alpha_1 \xi_1  \right)\Big|\xi_{2:K}\right]=\\
    \nonumber &=\int_{-\infty}^{\infty} \log_2(b_1 + \alpha_1 \xi_1) f_{\xi_1}(\xi_1) \mathrm{d}\xi_1 \\
    \nonumber  &= \frac{1}{2}\frac{1}{\ln(2)}\int_{0}^{\infty}  \ln(b_1 + \alpha_1 \xi_1) e^{-\frac{\xi_1}{2}} \mathrm{d}\xi_1\\
    \nonumber  &\overset{(a)}{=}\frac{1}{2}\frac{1}{\ln(2)}\Bigg(\left[ -2 e^{-\frac{\xi_1}{2}} \ln(b_1 + \alpha_1 \xi_1)  \right]_0^{\infty}\\
  \nonumber  &-\int_{0}^{\infty}  \frac{-2 \alpha_1e^{-\frac{\xi_1}{2}}}{b_1+\alpha_1\xi_1} \mathrm{d}\xi_1 \Bigg)\\
    \nonumber  &=\frac{1}{2}\frac{1}{\ln(2)}\left( 2 \ln(b_1)  + \int_{0}^{\infty}  \frac{e^{-\frac{\xi_1}{2}}}{\frac{b_1}{2 \alpha_1}+\frac{\xi_1}{2}} \mathrm{d}\xi_1 \right)\\
    \nonumber &\overset{(b)}{=}\frac{1}{2}\frac{1}{\ln(2)}\left( 2 \ln(b_1)  + \int_{\frac{b_1}{2\alpha_1}}^{\infty}  \frac{e^{-(u-\frac{b_1}{2\alpha_1}) }}{u} 2\mathrm{d}u \right)\\
 \nonumber   &=\frac{1}{2}\frac{1}{\ln(2)}\left( 2 \ln(b_1)  + 2e^{\frac{b_1}{2\alpha_1} }\int_{\frac{b_1}{2\alpha_1}}^{\infty}  \frac{e^{-u }}{u} \mathrm{d}u \right)\\
    &= \frac{1}{\ln(2)} \left( \ln(b_1) + e^{\frac{b_1}{2\alpha_1}}E_1\left( \frac{b_1}{2 \alpha_1}\right)\right)
\end{align}
where we used integration by parts in $(a)$ and integration by substitution in $(b)$.
In the last line, $E_1(x)$ is defined as the exponential integral.
The bounds for $E_1(x)$ in \cite[p. 229, eq. (5.1.20)]{Abramo} are not suited for our purposes and, hence, we derived a new lower bound for the exponential integral given in Lemma \ref{lemma:ExpIntegral}.
Using this lower bound, we obtain 
\begin{equation}\label{eq:ExpIntBoundApply}
    \begin{aligned}
    &\underset{\xi_1}{\mathbb{E}}\left[\log_2\left(b_1 + \alpha_1 \xi_1  \right)\Big|\xi_{2:K}\right]\\
    & > \frac{1}{\ln(2)} \left(\ln(b_1 ) + \ln\left(1 + e^{-\gamma} \frac{2 \alpha_1}{b_1}\right) \right)\\
    &=   \log_2(b_1 + 2  e^{-\gamma} \alpha_1)\\
\end{aligned} 
\end{equation}
where $\gamma$ is the Euler-Mascheroni constant.
By defining 
\begin{equation}
    b_k = 2  e^{-\gamma} \summe{j=1}{k-1}\alpha_j + \summe{j=k+1}{K}\alpha_j \xi_j
\end{equation}
and by iteratively proceeding with the same method for $k=2,\dots,K$ as we did above for $k=1$, we have in the $k$-th step
\begin{align}
       \nonumber & \underset{\xi_{k:K}}{\mathbb{E}}\left[\log_2(b_{k-1} + 2  e^{-\gamma}\alpha_{k-1})\right]\\
    \nonumber&= \underset{\xi_{k:K}}{\mathbb{E}}\left[\log_2(b_{k} + \xi_k \alpha_k)\right]\\
    \nonumber&=\underset{\xi_{k+1:K}}{\mathbb{E}}\left[\underset{\xi_k}{\mathbb{E}}\left[\log_2\left(b_{k} + \alpha_k \xi_k  \right)\Big|\xi_{k+1:K}\right]\right]\\
    &> \underset{\xi_{k+1:K}}{\mathbb{E}}\left[\log_2(b_{k} + 2  e^{-\gamma}\alpha_k)\right].
\end{align}
Hence, iteratively applying the bound in \eqref{eq:ExpIntBoundApply} for all $k=1,\dots,K$, we obtain in the final $K$th step the value 
$b_{K} =  2 e^{-\gamma}\summe{k=1}{K-1}\alpha_k$ and, therefore, we can bound the expression in \eqref{eq:LawOfTotalFirststep} as 
\begin{align}\label{eq:MultiUserBound}
    \underset{\xi_{1:K}}{\mathbb{E}}\left[\log_2\left( \summe{k=1}{K}\alpha_k \xi_k  \right) \right] 
    >\log_2\left( 2   e^{-\gamma}\summe{k=1}{K}\alpha_k\right).
\end{align}
This bound is valid for $K \ge 2$. When $K=1$ the situation is different and we can directly obtain an expression for the expression in \eqref{eq:LowerBoundXi}, given by 
\begin{equation}
    \expct{\text{Mit}} \ge \underset{\bthet}{\mathbb{E}} \left[ \underset{\xi_1}{\mathbb{E}}\left[\log_2 \left(\alpha_1 \xi_1   \right)\right]\right]=\underset{\bthet}{\mathbb{E}} \left[ \log_2 ( \alpha_1) +  \underset{\xi_1}{\mathbb{E}}\left[\log_2 \left( \xi_1   \right)\right]\right].
\end{equation}
From \cite[p. 943, eq. (26.4.36)]{Abramo}, we know that 
\begin{equation}
    \underset{\xi_1}{\mathbb{E}}\left[\log_2 \left( \xi_1   \right)\right] = \log_2(2) + \underset{\xi_1}{\mathbb{E}}\left[\log_2 \left( \frac{\xi_1}{2}   \right)\right]= \log_2(2)+ \psi(1)
\end{equation}
where $\psi(\bullet)$ is the digamma function. Using \cite[p.258, eq. (6.3.2)]{Abramo}, we have $\psi(1) = -\gamma$ and, hence, we arrive at
\begin{equation}\label{eq:SingleUserBound}
    \underset{\xi_1}{\mathbb{E}}\left[\log_2 \left(\alpha_1 \xi_1   \right)\right]= \log_2(2 e^{-\gamma}\alpha_1).
\end{equation}
Combining the bound from \eqref{eq:MultiUserBound} for $K\ge 2$ and the expression for $K=1$ in \eqref{eq:SingleUserBound}, we obtain the bound for equation \eqref{eq:Upperboundwithh} with the help of \eqref{eq:LowerBoundXi} as
\begin{align}
    \nonumber \expct{\text{Mit}}&\ge \underset{\bh_{\tre,K+1}, \bthet}{\mathbb{E}}\left[\underset{\bH_\tre^\ts|\bh_{\tre,K+1}, \bthet}{\mathbb{E}}\left[ \log_2\left( \summe{k=1}{K} \frac{ \abs{h_k}^2 }{\tr(\bm{R}_{\td,k})}  \right) \right]\right]\\
    \nonumber&= \underset{\bh_{\tre,K+1}, \bthet}{\mathbb{E}}\left[\underset{\bH_\tre^\ts|\bh_{\tre,K+1}, \bthet}{\mathbb{E}}\left[ \log_2\left( \summe{k=1}{K} \alpha_k \xi_k \right)\right] \right]\\
    \nonumber &\ge \underset{\bh_{\tre,K+1}, \bthet}{\mathbb{E}}\left[\log_2\left(2  e^{-\gamma} \summe{k=1}{K}\alpha_k\right)\right]\\
    &= \underset{\bh_{\tre,K+1}, \bthet}{\mathbb{E}}\left[\log_2\left(  e^{-\gamma} \summe{k=1}{K}\frac{\bthet^\Her \bm{R}_{\tc,k} \bthet}{\tr(\bm{R}_{\td,k})}  \right)\right]
\end{align}
and the final upper bound of the ergodic rate is given by 
\begin{equation}\label{eq:ErgLinearFinalRateBound}
    \expct{\overline{\ac{SE}}_{\text{Lin},\tre}} \le  \expct{\log_2\left(\frac{\abs{\bh_{\tc,K+1}^\Her \bthet}^2\bar{p}}{   e^{-\gamma}\summe{k=1}{K}\frac{\bthet^\Her \bm{R}_{\tc,k} \bthet}{\tr(\bm{R}_{\td,k})}  }\right)}.\\
\end{equation}

\subsection{SE Expressions for Random/Statistical Phase Shifts under i.i.d. Rayleigh Fading}\label{app:LinearDPCUpperBoundErgodicRateiidRayleighRandom}
Under Rayleigh fading, it is possible to derive an exact expression for
\begin{equation}
     \expct{\log_2\left( \abs{\bh_{\tc,K+1}^\Her \bthet}^2 \bar{p} \right)}.
\end{equation}
We start similar to Appendix \ref{app:LinearUpperBoundErgodicRate} by defining
\begin{equation}
    \xi = \frac{2}{\text{var}[{\bh_{\tc,K+1}^\Her\bthet}]} \abs{\bh_{\tc,K+1}^\Her \bthet}^2 \; \sim \chi^2(2)
\end{equation}
which is chi-squared distributed with two degrees of freedom and the variance $\text{var}[{\bh_{\tc,K+1}^\Her\bthet}]$ is given by
\begin{equation}
    \text{var}[{\bh_{\tc,K+1}^\Her\bthet}] = \expct{\abs{\bh_{\tc,K+1}^\Her\bthet}^2} = \bthet^\Her \bm{R}_{\tc,K+1} \bthet.
\end{equation}
Using \eqref{eq:SingleUserBound} according to \cite[p. 258, eq. (6.3.2)]{Abramo}, we have 
\begin{align}\label{eq:StatPhaseNumeratorLinearRayleigh}
      \nonumber   \expct{\log_2\left( \abs{\bh_{\tc,K+1}^\Her \bthet}^2 \bar{p} \right)}  &=  \expct{\log_2\left( \frac{1}{2} \bthet^\Her \bm{R}_{\tc,K+1} \bthet \xi \bar{p} \right)}\\
        &=  {\log_2\left(  e^{-\gamma} \bthet^\Her \bm{R}_{\tc,K+1} \bthet  \bar{p} \right)}
\end{align}
for statistical/random phase shifts.
Under i.i.d. Rayleigh fading, this reduces to 
\begin{equation}
    \expct{\log_2\left( \abs{\bh_{\tc,K+1}^\Her \bthet}^2 \bar{p} \right)} = \log_2(e^{-\gamma}L_{\tre,K+1} L_{\text{G}} \NB \NRe \bar{p})
\end{equation}
and we arrive at the \ac{SE} for \ac{DPC}
\begin{align}\label{eq:ErgodicRate}
    \nonumber  \expct{\overline{\text{SE}}_{\text{DPC},\tre}} = \log_2(e^{-\gamma}L_{\tre,K+1} L_{\text{G}} \NB \NRe \bar{p}).
  \end{align}
  as well as the upper bound for linear precoding
    \begin{align}
     \nonumber   \expct{\overline{\ac{SE}}_{\text{Lin},\tre}} &\le  \log_2\left( \frac{ e^{-\gamma}L_{\tre,K+1}  L_{\text{G}} \NB\NRe \bar{p} }{   e^{-\gamma}\frac{\NRe L_{\text{G}} \NB}{\NB} \summe{k=1}{K} \frac{L_{\tre,k}}{L_{\td,k}} }\right)\\
    &= \log_2\left( \NB \frac{L_{\tre,K+1}   \bar{p} }{  \summe{k=1}{K} \frac{L_{\tre,k}}{L_{\td,k}} }\right).
\end{align}

\subsection{Linear Precoding Upper Bound for Weak User Maximization under i.i.d. Rayleigh Fading}\label{app:LinearUpperBoundErgodicRateiidRayleighInst}
When considering the phases based on optimizing the weak user, i.e., 
\begin{equation}
    \bthet = \exp(\mathrm{j}\arg(\bh_{\tc,K+1}))
\end{equation}
we arrive at 
\begin{align}
     \nonumber &\expct{\log_2\left( \abs{\bh_{\tc,K+1}^\Her \bthet}^2 \bar{p} \right)} 
     \le \log_2\left( \expct{\abs{\bh_{\tc,K+1}^\Her \bthet}^2} \bar{p} \right)\\
     &=\log_2\left( \frac{\pi}{4} \NB L_{\text{G}}L_{\tre,K+1} \NRe^2\bar{p} \right)
 \end{align}
according to \cite{SNRRayleigh}.
Hence, the upper bound in case of optimizing the weak user results in
\begin{equation}
    \begin{aligned}
        \expct{\overline{\ac{SE}}_{\text{Lin},\tre}} &\le  \log_2\left( \frac{\pi e^{\gamma}}{4}  \NRe \NB\frac{L_{\tre,K+1}   \bar{p} }{  \summe{k=1}{K} \frac{L_{\tre,k}}{L_{\td,k}} }\right).
\end{aligned}
\end{equation}

\subsection{DPC Lower Bound for Weak User Maximization under i.i.d. Rayleigh Fading}\label{app:DPCUpperBoundsInst}
We are considering i.i.d. Rayleigh Fading, i.e. \mbox{$\bh_{\tre,K+1} \sim \gaussdist{\bm{0}}{\eye\; L_{\tre,K+1}}$}
as well as the optimal phase shifts based on instantaneous \ac{CSI} given by alignment as
\begin{equation}
    \bthet = \exp(\mathrm{j}\arg(\bh_{\tc,K+1})).
\end{equation}
In this case, we have
    \begin{align}\label{eq:DPCRateOptimalPhases}
      \nonumber  \expct{\overline{\text{SE}}_{\text{DPC},\tre}} &= \expct{\log_2\left(\abs{\bh_{\tc,K+1}^\Her \bthet}^2 \bar{p} \right)}\\
        &= \log_2\left( \bar{p} \right) + 2\expct{ \log_2\left(\summe{n=1}{\NRe}\abs{h_{\tc,K+1,n}}\right) }.
    \end{align}
Defining 
\begin{equation}
    \zeta_n = \sqrt{\frac{2}{\NB L_{\tre,K+1} L_{\text{G}}}} \abs{h_{\tre,K+1,n}},
\end{equation}
we obtain with Jensen's inequality
\begin{equation}
    \begin{aligned}
        &\expct{\overline{\text{SE}}_{\text{DPC},\tre}} =\log_2\left(  \bar{p} \right) + 2\expct{ \log_2\left(\sqrt{\frac{\NB L_{\tre,K+1} L_{\text{G}}}{2}}\summe{n=1}{\NRe}\zeta_n\right) }\\
        &= \log_2\left( \frac{1}{2}\NB L_{\tre,K+1} L_{\text{G}} \bar{p} \NRe^2 \right)+ 2\expct{ \log_2\left(\frac{1}{\NRe}\summe{n=1}{\NRe}\zeta_n\right) }\\
        &\ge \log_2\left(  \frac{1}{2}L^\prime_G \bar{p}\NB \NRe^2 L_{\tre,K+1}\right) + \frac{2}{\NRe}\expct{ \summe{n=1}{\NRe}\log_2\left( \zeta_n\right)  }.\\
    \end{aligned}
\end{equation}
By recognizing that $\zeta_n^2$ is again chi-squared distributed with two degrees of freedom, we can use the same argumentation as in Eqn. \eqref{eq:SingleUserBound} based on \cite[p. 258, eq. (6.3.2)]{Abramo} , after which we arrive at
\begin{equation}
    \expct{\log_2( \zeta_n) }= \frac{1}{2}\expct{\log_2(\zeta_n^2)  } =\frac{1}{2} \log_2(2e^{-\gamma})
\end{equation}
and obtain the lower bound
\begin{equation}
    \expct{\overline{\text{SE}}_{\text{DPC},\tre}} \ge \log_2\left( e^{-\gamma} L^\prime_G L_{\tre,K+1}\bar{p} \NB \NRe^2 \right)
\end{equation}
for the optimal phase shifts with instantaneous \ac{CSI} under i.i.d. Rayleigh fading.

\bibliographystyle{IEEEtran}
\bibliography{refs}

\end{document}